\documentclass[12pt,leqno]{article}
\usepackage[fleqn]{amsmath}
\usepackage{amssymb,latexsym,color,dsfont,enumitem}
\usepackage[margin=1.1in]{geometry}
\usepackage[auth-sc]{authblk}
\usepackage{tikz}
\usepackage{pgf}
\usepackage{wasysym} 
\usepackage{mathtools}

\usepackage[cp1251]{inputenc}
\usepackage[T2A]{fontenc}
\usepackage[english,russian]{babel}

\newtheorem{theorem}{Theorem}[section]

\newtheorem{lemma}[theorem]{Lemma}

\newtheorem{corollary}[theorem]{Corollary}

\newenvironment{proof}{\noindent
  \textbf{Proof.}}{\hfill$\Box$\\}

\makeatletter
\newcommand{\ExtDiamond}{\@ifnextchar({\ExtDiamond@i}{\ExtDiamond@i({p})}}
\def\ExtDiamond@i(#1){{\Diamond\!\!\!\!\Diamond}_{#1}}
\makeatother

\newcommand{\mynext}{\!\raisebox{-.2ex}{ 
            \mbox{\unitlength=0.9ex
            \begin{picture}(2,2)
            \linethickness{0.06ex}
            \put(1,1){\circle{2}} 
            \end{picture}}}       
            \,}

\renewcommand{\iff}{\Longleftrightarrow}

\newlength{\templength}
\newcommand{\doublesymbol}[2]
{
\settowidth{\templength}{\mbox{#1}}
\mbox{#1}\hspace{-\templength}\hspace{#2}\mbox{#1}
}
\newcommand{\doubleDiamondOne}{\doublesymbol{$\Diamond$}{0.10em}}
\newcommand{\doubleDiamondTwo}{\doublesymbol{$\Diamond$}{0.16em}}
\newcommand{\pDiamond}{\doubleDiamondOne}
\newcommand{\PDiamond}{\doubleDiamondTwo}
\newcommand{\nat}{\ensuremath{\mathds{N}}}
\newcommand{\con}{\wedge}

\newcommand{\dis}{\vee}

\newcommand{\imp}{\rightarrow}
\newcommand{\equivalence}{\leftrightarrow}

\newcommand{\vp}{\ensuremath{\varphi}}

\begin{document}

\selectlanguage{english}

\sloppy

\title{Algorithmic properties of first-order modal logics of linear
  Kripke frames in restricted languages\thanks{To appear in {\em Journal of Logic and Computation},
    \texttt{https://doi.org/10.1093/logcom/exab030}}} \author[1]{Mikhail Rybakov}
\author[2]{Dmitry Shkatov} \affil[1]{Institute for Information
  Transmission Problems of Russian Academy of Sciences, Moscow 127051,
  Russia} \affil[2]{School of Computer Science and Applied
  Mathematics, University of the Witwatersrand, Johannesburg,
  WITS2050, South Africa, \texttt{shkatov@gmail.com}} \date{\today}
\maketitle

\begin{abstract}
  We study the algorithmic properties of first-order monomodal logics
  of frames $\langle \nat, \leqslant \rangle$,
  $\langle \nat, < \rangle$, $\langle \mathds{Q}, \leqslant \rangle$,
  $\langle \mathds{Q}, < \rangle$,
  $\langle \mathds{R}, \leqslant \rangle$,
  $\langle \mathds{R}, < \rangle$, as well as some related logics, in
  languages with restrictions on the number of individual variables as
  well as the number and arity of predicate letters.  We show that the
  logics of frames based on $\nat$ are $\Pi^1_1$-hard---thus, not
  recursively enumerable---in languages with two individual variables,
  one monadic predicate letter and one proposition letter.  We also
  show that the logics of frames based on $\mathds{Q}$ and
  $\mathds{R}$ are $\Sigma^0_1$-hard in languages with the same
  restrictions.  Similar results are obtained for a number of related
  logics.
\end{abstract}

\section{Introduction}

How algorithmically expressive are first-order modal logics?  More
expressive, it is reasonable to assume, than the classical first-order
logic $\mathbf{QCl}$---just as propositional modal logics are, as a
rule, computationally harder than the classical propositional logic.
(In this context, it is natural to consider logics as sets of
validities, rather than as calculi: understood as calculi
conservatively extending $\mathbf{QCl}$ with a recursively enumerable
set of axioms and finitary rules of inference, first-order modal
logics are {\em a priori} $\Sigma^0_1$-complete.\footnote{The reader
  in need of a reminder of the basic concepts of computability theory
  may consult~\cite{Rogers}.})  Numerous first-order modal logics are,
however, just as algorithmically expressive as $\mathbf{QCl}$, i.e.
$\Sigma^0_1$-complete: some---such as $\mathbf{ QK}$, $\mathbf{QT}$,
$\mathbf{QD}$, $\mathbf{QKB}$, $\mathbf{QKTB}$, $\mathbf{QS4}$ and
$\mathbf{QS5}$---are recursively axiomatizable over $\mathbf{QCl}$;
others---such as logics of elementary classes of Kripke frames---are
recursively embeddable~\cite{RSh20JLC,RShAiML18} into $\mathbf{QCl}$
by the standard translation~\cite{Benthem85,ChRyb00}, \cite[Section
3.12]{GShS}.  This suggests a need for a more fine-grained analysis,
which takes into account the algorithmic expressivity not only of full
logics but also of their fragments obtained by placing restrictions on
the structure of formulas.  Such analysis allows us to distinguish
$\Sigma^0_1$-complete modal logics from $\mathbf{QCl}$ by algorithmic
expressivity: while the monadic fragment of $\mathbf{QCl}$ is
decidable~\cite{Lowenheim15,Behmann22}, the monadic fragments of most
$\Sigma^0_1$-complete modal logics are not~\cite{Kripke62}; while the
two-variable fragment of $\mathbf{QCl}$ is
decidable~\cite{Mortimer75,GKV97}, the two-variable fragments of most
$\Sigma^0_1$-complete modal logics are not~\cite{KKZ05}.  This leads
to the study of the algorithmic properties of the {\em fragments} of
first-order modal logics.

This study is also motivated by an underdevelopment, relative
to~$\mathbf{QCl}$~\cite{BGG97}, of the {\em algorithmic classification
  problem} for first-order modal logics---an effort to identify their
maximal decidable and minimal undecidable fragments.  Despite
extensive
literature~\cite{Kripke62,MMO65,Mints68,Ono77,Gabbay81,AD90,GSh93,WZ01,KKZ05,RSh19SL,RShJLC20a},
whose summary can be found in the Introduction to the authors' earlier
article~\cite{RShJLC20a}, much less is known about the algorithmic
properties of the fragments of first-order modal, and closely related
superintuitionistic, logics than about the algorithmic properties of
the fragments of $\mathbf{QCl}$.

The algorithmic properties of one-variable and two-variable fragments
of first-order modal logics are also of interest due to close links
between those fragments and, respectively, two-dimensional and
three-dimensional propositional modal
logics~\cite{GSh98,GKWZ,Shehtman19}.

In the present paper, we attempt to identify the minimal undecidable
fragments of the first-order monomodal logics of frames
$\langle \nat, \leqslant \rangle$, $\langle \nat, < \rangle$,
$\langle \mathds{Q}, \leqslant \rangle$,
$\langle \mathds{Q}, < \rangle$,
$\langle \mathds{R}, \leqslant \rangle$ and
$\langle \mathds{R}, < \rangle$, as well as of closely related linear
orders.\footnote{Preliminary results on the logics of
  $\langle \nat, \leqslant \rangle$ and $\langle \nat, < \rangle$ were
  reported in a conference paper~\cite{RSh20AiML}.  The present
  article improves on the earlier paper in two respects.  First, we
  obtain stronger results on the logics of frames based on the
  naturals by proving $\Sigma^1_1$-hardness for weaker languages; the
  results reported in this article are plausibly optimal, as discussed
  in Section~\ref{sec:discussion}.  Second, we report results on
  logics of the rationals (and hence on $\mathbf{QS4.3}$,
  $\mathbf{QK4.3.D.X}$ and $\mathbf{QK4.3}$), the reals and infinite
  ordinals distinct from $\omega$.}  The logics of these structures
are of interest on at least three counts.

First, the structures themselves are of interest, for at least two
reasons.  They have long been considered natural models of the flow of
time~\cite{Prior67, Goldblatt92,GHR94}; therefore, their study has
been stimulated by the long-standing interest in temporal reasoning.
Even though we focus on monomodal languages, since our results are
negative, they do apply to more expressive languages with modalities
for the past as well as the future.  On a more basic level perhaps,
the structures based on the naturals, the rationals and the reals are
so fundamental to mathematics that the properties of the corresponding
logics are of intrinsic mathematical significance: the classical
theories of these structures, both first-order and second-order, have
been extensively studied; in particular, it has long been known that
the monadic second-order theory of $\langle \nat, < \rangle$ is
decidable~\cite{Elgot61}.

Second, the logics considered here call for techniques substantially
different from those used in the previous
studies~\cite{Kripke62,GSh93,KKZ05,RSh19SL,RShJLC20a} of the
algorithmic properties of fragments of monomodal predicate logics.
For most such logics, known undecidability proofs for fragments with a
single monadic predicate letter, a restriction considered here,
require a transformation of models that increases their branching
factor.  The methods the authors used earlier~\cite{RSh19SL,RShJLC20a}
intrinsically rely on increasing the branching factor of models, a
feature inherited from the propositional-level
techniques~\cite{Halpern95,ChRyb03,RShIGPL18,RShICTAC18,RShIGPL19,RShJLC21b}
those methods are based on.  On the other hand, the construction by
Blackburn and Spaan~\cite{BS93}, which is propositional but, in
principle, adaptable to first-order logics, does not seem to be
readily applicable to logics of transitive frames since it relies on
the use of a modal operator suitable for counting transitions along
the accessibility relation of a frame.  The techniques used here
should, therefore, be of relevance to the study of the algorithmic
properties of fragments of monomodal logics of structures with a
resticted branching factor, including trees.

Third, the logics of frames $\langle \nat, \leqslant \rangle$ and
$\langle \nat, < \rangle$ are algorithmically hard---as follows from
Theorem~\ref{thr:nat} below, they are $\Pi^1_1$-hard.  Most research
into the algorithmic properties of monomodal, and closely related
superintuitionistic, predicate logics has been focused on decidability
and undecidability. The only study to date~\cite{RShJLC20a}, as far as
we know, of the algorithmic properties of fragments of not recursively
enumerable monomodal predicate logics concerns logics of frames with
finite sets of worlds (likewise, very few results~\cite[Theorem 1,
p. 272]{Gabbay81}, \cite{RShJLC21a} are known on algorithmic
properties of fragments of not recursively enumerable
superintuitionistic predicate logics).  While it is natural that
decidability and undecidability are the main concern of the Classical
Decision Problem~\cite{BGG97}, the study of the algorithmic properties
of modal first-order logics should, we believe, involve identifying
minimal fragments that are as hard---in pertinent classes of the
arithmetical, or the analytical, hierarchy---as the full logics.

The algorithmic properties of fragments of not recursively enumerable
logics have been, however, extensively studied in the context of
first-order languages more expressive than monomodal ones considered
here---most recently, by Hodkinson, Wolter and
Zakharyaschev~\cite{HWZ00,WZ01} (for a summary, see \cite[Chapter
11]{GKWZ}; for earlier work,
see~\cite{ANS79,Szalas86,SzalasHolenderski88,Abadi89,Merz92}).  The
methods used here have been inspired by those of Wolter and
Zakharyaschev~\cite[Theorem 2.3]{WZ01}, who encode a $\Sigma^1_1$-hard
tiling problem in a first-order language with two modal operators, one
corresponding to a basic accessibility relation and the other to its
reflexive transitive closure.  A similar result~\cite[Theorem
2]{HWZ00} has been obtained by Hodkinson, Wolter and Zakharyaschev for
the first-order temporal logic of $\langle \nat, \leqslant \rangle$ in
the language with two temporal operators: ``next,'' corresponding to
the immediate successor relation on $\nat$, and ``always in the
future,'' corresponding to its reflexive transitive closure, the
partial order $\leqslant$ (both operators can be expressed with a
binary temporal operator ``until'').\footnote{We touch on such
  languages in Section~\ref{sec:discussion}.}  Another similar
result~\cite[Theorem 5.6]{Wolter00} has been obtained by Wolter for
the first-order logics containing, alongside the individual knowledge
operators, the common knowledge operator whose semantics involves the
reflexive transitive closure of the union of the accessibility
relations for the individual knowledge operators.  We extend herein to
monomodal logics, which do not have expressive power for capturing the
reflexive transitive closures of accessibility relations, techniques
developed by Hodkinson, Wolter, and
Zakharyaschev~\cite{HWZ00,Wolter00,WZ01}.

The paper is structured as follows.  In Section~\ref{sec:prelim}, we
introduce preliminaries on first-order modal logic. In
Section~\ref{sec:ref}, we prove that satisfiability for the logic of
$\langle \nat, \leqslant \rangle$ is $\Sigma^1_1$-hard in languages
with two individual variables, a single monadic predicate letter and a
single proposition letter.  In Section~\ref{sec:discrete}, the results
of Section~\ref{sec:ref} are extended to logics of frames
$\langle \nat, R \rangle$, where $R$ is a binary relation between $<$
and $\leqslant$, and to frames based on infinite ordinals of a special
form.  In Section~\ref{sec:dense}, we prove, by modifying the argument
of Sections~\ref{sec:ref} and~\ref{sec:discrete}, that satisfiability
for logics of $\langle \mathds{Q}, \leqslant \rangle$,
$\langle \mathds{Q}, < \rangle$,
$\langle \mathds{R}, \leqslant \rangle$ and
$\langle \mathds{R}, < \rangle$ is $\Pi^0_1$-hard in languages with
the same restrictions.  In Section~\ref{sec:rest}, we briefly mention
some corollaries of the results proven earlier.  We conclude, in
Section~\ref{sec:discussion}, by discussing first-order temporal
logics with modalities ``next'' and ``always in the future,'' as well
as questions for future study.

\section{Preliminaries}
\label{sec:prelim}

An unrestricted first-order predicate modal language contains
countably many individual variables; countably many predicate letters
of every arity, including 0 (nullary predicate letters are also called
proposition letters); the propositional constant $\bot$ (falsity), the
binary propositional connective $\imp$, the unary modal connective
$\Box$ and the quantifier $\forall$.  Formulas as well as the symbols
$\top$, $\neg$, $\vee$, $\wedge$, $\equivalence$, $\exists$ and
$\Diamond$ are defined in the usual way.  We also use the
abbreviations $\Box^0 \vp = \vp$, $\Box^{n+1} \vp = \Box \Box^n \vp$
and $\Diamond^n \vp = \neg \Box^n \neg \vp$, for every $n \in \nat$.

When parentheses are omitted, unary connectives and quantifiers are
assumed to bind tighter than $\wedge$ and $\vee$, which are assumed to
bind tighter than $\imp$ and $\equivalence$. We usually write atomic
formulas in prefix notation; for some predicate letters we, however,
use infix.

A {\em normal predicate modal logic\/} is a set of formulas containing
the validities of the classical first-order predicate logic
$\mathbf{QCl}$, as well as the formulas of the form
$\Box (\vp \imp \psi) \imp (\Box \vp \imp \Box \psi)$, and closed
under predicate substitution, modus ponens, generalisation and
necessitation.\footnote{The reader wishing a reminder of the
  definition of these closure conditions may consult~\cite[Definition
  2.6.1]{GShS}; for a detailed discussion of predicate substitution,
  consult~\cite[\S2.3, \S2.5]{GShS}.}

In this paper, we are interested in predicate logics defined using the
Kripke semantics.\footnote{For Kripke semantics for predicate modal
  logics, see \cite{ShS90,SSh93,HC96,FM98,Garson01,BG07,Goldblatt11},
  \cite[\S 3.1]{GShS}.}

A {\em Kripke frame\/} is a tuple $\frak{F} = \langle W,R\rangle$,
where $W$ is a non-empty set of {\em possible worlds\/} and $R$ is a
binary {\em accessibility relation\/} on $W$; if $w R v$, we say that
$v$ {\em is accessible from\/} $w$ and that $w$ {\em sees\/} $v$.

A {\em predicate Kripke frame with expanding domains\/} is a tuple
$\frak{F}_D = \langle W,R, D\rangle$, where $\langle W,R\rangle$ is a
Kripke frame and $D$ is a function from $W$ into the set of non-empty
subsets of some set, {\em the domain of\/~$\frak{F}_D$}; the function
$D$ is required to satisfy the condition that $wRw'$ implies
$D(w) \subseteq D(w')$. The set $D(w)$, also denoted by $D_w$, is {\em
  the domain of\/~$w$}.  We also consider predicate frames satisfying
the stronger condition that $D(w) = D(w')$, for every $w, w' \in W$;
such frames are {\em predicate frames with a constant
  domain}.\footnote{More precisely, such predicate frames are known as
  predicate frames with globally constant domains.  For connected
  predicate frames, the global constancy condition given above is
  equivalent to the local constancy condition requiring that
  $D(w) = D(w')$ whenever $w R w'$.  Since the frames we consider are
  rooted, and therefore connected, the distinction between global and
  local constancy is immaterial for the purposes of this paper.}  {\em
  Predicate frame\/} simpliciter means a predicate frame with
expanding domains.

A {\em Kripke model\/} is a tuple $\frak{M} = \langle W,R,D,I\rangle$,
where $\langle W,R, D\rangle$ is a predicate Kripke frame and $I$,
{\em the interpretation of predicate letters\/} with respect to worlds
in $W$, is a function assigning to a world $w\in W$ and an $n$-ary
predicate letter $P$ an $n$-ary relation $I(w,P)$ on $D(w)$---i.e.,
$I(w,P) \subseteq D_w^n$.  In particular, if $p$ is a proposition
letter, then $I(w,P) \subseteq D_w^0 = \{ \langle \rangle \}$; thus,
we can identify truth with $\{ \langle \rangle \}$ and falsity with
$\varnothing$.  We often write $P^{I, w}$ instead of $I(w,P)$. We say
that a model $\langle W,R,D,I\rangle$ is {\em based on\/} the frame
$\langle W,R \rangle$ and is {\em based on\/} the predicate frame
$\langle W,R,D\rangle$.

We use the standard notation for binary relations: the $n$-fold, for
each $n \in \nat^+$, composition of a binary relation $R$ is denoted
by $R^n$; if $R$ is a binary relation on a non-empty set $W$ and
$w \in W$, then $R(w) = \{ v \in W : w R v \}$.

An {\em assignment\/} in a model is a function $g$ associating with
every individual variable $x$ an element $g(x)$ of the domain of the
underlying predicate frame.  We write $g' \stackrel{x}{=} g$ if
assignment $g'$ differs from assignment $g$ in at most the value of
$x$.

The truth of a formula $\varphi$ at a world $w$ of a model $\frak{M}$
under an assignment $g$ is defined recursively:
\begin{itemize}
\item $\frak{M},w\models^g P(x_1,\ldots,x_n)$ if
  $\langle g(x_1),\ldots,g(x_n)\rangle \in P^{I, w}$, where $P$ is an
  $n$-ary predicate letter;
\item $\frak{M},w \not\models^g \bot$;
\item $\frak{M},w\models^g\varphi_1 \imp \varphi_2$ if
  $\frak{M},w\models^g\varphi_1$ implies $\frak{M},w\models^g\varphi_2$;
\item $\frak{M},w\models^g\Box\vp_1$ if
  $\frak{M},w'\models^g\varphi_1$, for every $w' \in R(w)$;
\item $\frak{M},w\models^g\forall x\,\varphi_1$ if
  $\frak{M},w\models^{g'}\varphi_1$, for every $g'$ such that
  $g' \stackrel{x}{=} g$ and $g'(x)\in D_w$.
\end{itemize}

Observe that, if $\frak{M} = \langle W,R,D,I\rangle$ is a Kripke
model, $w \in W$ and $I_w(P) = I(w,P)$, then
$\frak{M}_w = \langle D_w, I_w \rangle$ is a classical model, or
structure.

We shall often use the following notation.  Let
$\frak{M} = \langle W,R,D,I\rangle$ be a model, $w \in W$, and
$a_1, \ldots, a_n \in D_w$; let also $\vp(x_1, \ldots, x_n)$ be a
formula whose free variables are among $x_1, \ldots, x_n$ and $g$ an
assignment with $g(x_1) = a_1, \ldots, g(x_n) = a_n$.  Then, we write
$\frak{M}, w \models \vp (a_1, \ldots, a_n)$ instead of
$\frak{M}, w \models^g \vp (x_1, \ldots, x_n)$.  This notation is
unambiguous since the languages we consider lack constants and the
truth value of $\vp(x_1, \ldots, x_n)$ does not depend on the values
of variables other than $x_1, \ldots, x_n$.

A formula $\vp$ is {\em true at a world\/} $w$ of a model $\frak{M}$
(in symbols, $\frak{M},w\models \vp$, or simply $w \models \vp$ if
$\frak{M}$ is clear from the context) if $\frak{M},w\models^g \vp$,
for every $g$ assigning to free variables of $\vp$ elements of $D_w$.
A formula $\vp$ is {\em true in a model\/~$\frak{M}$} (in symbols,
$\frak{M} \models \vp$) if $\frak{M},w\models \vp$, for every world
$w$ of $\frak{M}$.  A formula $\vp$ is {\em valid on a predicate
  frame\/~$\frak{F}_D$} if $\vp$ is true in every model based on
$\frak{F}_D$. A formula $\vp$ is {\em valid on a frame\/~$\frak{F}$}
(in symbols, $\frak{F} \models \vp$) if $\vp$ is valid on every
predicate frame $\langle \frak{F}, D \rangle$.  These notions, and the
corresponding notation, can be extended to sets of formulas, in a
natural way.

We shall often rely on the following observation: if a model
$\frak{M}$ is based on a predicate frame with a constant domain, then
$\frak{M},w\models \vp$ if, and only if, $\frak{M},w\models^g \vp$,
for every assignment $g$.

Let $\frak{C}$ be a class of Kripke frames.  The set of formulas valid
on every frame in $\frak{C}$ is a predicate modal logic, which we
denote by $\mathbf{L} (\frak{C})$; we write $\mathbf{L} (W, R)$
instead of $\mathbf{L} (\{\langle W, R \rangle \})$.  The set of
formulas valid on every predicate frame with a constant domain based
on some frame in $\frak{C}$ also is a predicate modal logic, which we
denote by ${\mathbf{L}_{c}} (\frak{C})$; we write
${\mathbf{L}_{c}} (W, R)$ instead of
${\mathbf{L}_{c}} (\{\langle W, R \rangle \})$.

\section{The first-order logic of $\langle \nat, \leqslant \rangle$}
\label{sec:ref}

In this section, we prove that satisfiability for
$\mathbf{L}( \nat, \leqslant )$ is $\Sigma^1_1$-hard---hence,
$\mathbf{L}(\nat, \leqslant )$ is $\Pi^1_1$-hard, and therefore not
recursively enumerable---in languages with two individual variables,
one monadic predicate letter and one proposition letter.

\subsection{Reduction from a tiling problem}
\label{sec:reduction}

We do so by encoding the following $\Sigma^1_1$-complete~\cite[Theorem
6.4]{Harel86} $\nat \times \nat$ recurrent tiling problem.  We are
given a set of {\em tiles}, a tile $t$ being a $1 \times 1$ square,
with a fixed orientation, whose edges are colored with
$\textit{left}(t)$, $right(t)$, $up(t)$ and $down(t)$.  A {\em tile
  type} is a quadruplet of edge colors.  Each tile has a type from the
set $T = \{t_0, \ldots, t_s \}$, tiles of each type being in an
unlimited supply. A {\em tiling\/} is an arrangement of tiles on the
rectangular $\nat \times \nat$ grid so that the edge colors of the
adjacent tiles match, both horizontally and vertically.  We are to
determine whether there exists a tiling of the grid in which a tile of
type $t_0$ occurs infinitely often in the leftmost column, i.e.,
whether there exists a function $f: \nat \times \nat \to T$ such that,
for every $n, m \in \nat$,
\begin{itemize}
\item[]($T_1$)~ $right(f(n,m)) = \textit{left} (f(n+1,m))$;
\item[]($T_2$)~ $up(f(n,m)) = \textit{down} (f(n,m+1))$;
\item[]($T_3$)~ the set $\{ m \in \nat : f(0, m) = t_0 \}$ is
  infinite.
\end{itemize}

The idea of the encoding we use is based on the work of Hodkinson,
Wolter and Zakharyaschev~\cite[Theorem 2]{HWZ00} (also
see~\cite[Theorem 11.1]{GKWZ}; similar constructions have been used
elsewhere~\cite{Spaan93,Marx99,WZ01,KKZ05}), but the encoding itself
is more involved since our language lacks the ``next'' operator
available to them (we touch on languages with ``next'' in
Section~\ref{sec:discussion}).  To make the underlying idea clearer,
we construct, in the initial encoding, a formula of two individual
variables without regard for the number of predicate letters involved;
subsequently, we reduce the formula thus obtained to a formula with a
single monadic and a single proposition letter.

Let $\triangleleft$ be a binary predicate letter, $M$ and $P_t$---for
every $t \in T$---monadic predicate letters and $p$ a proposition
letter.

Given a formula $\vp$ in such a language, define
$$
\begin{array}{c}
  \pDiamond\varphi = \Diamond(p \con\Diamond( \neg p \con\varphi));
  \medskip\\
  \pDiamond^0 \varphi = \varphi; \quad \pDiamond^{n+1} \varphi =
  \pDiamond\pDiamond^n\varphi, \mbox{ for every } n \in \nat.
\end{array}
$$
The operator $\pDiamond$ forces a transition to a different world when
evaluating a formula $\pDiamond \vp$ in a reflexive model, just as
$\Diamond$ does in an irreflexive one. To make this explicit, we
define, given a model based on the frame
$\langle \nat, \leqslant \rangle$, a binary relation
$R_{\scriptsize\pDiamond}$ on $\nat$ by
$$
\begin{array}{lcl}
w R_{\scriptsize\pDiamond} v
  & \leftrightharpoons
  & \mbox{$v \not\models p$ and, for some $u \in \nat$, both $w \leqslant u \leqslant v$
          and $u \models p$.}
\end{array}
$$
Thus, $R_{\scriptsize\pDiamond}$ is irreflexive and transitive.

We also define
$$
\begin{array}{lcl}
  U (x) & = & \displaystyle\bigwedge\limits_{t \in T} \neg P_t(x).
\end{array}
$$

In such a language, define (for brevity, in formulas we write $l$,
$r$, $u$ and $d$ instead of \textit{left}, \textit{right},
\textit{up} and \textit{down})
$$
\begin{array}{rcl}
  A_0 & = & \exists x\,  \Box\, U (x); \\
  A_1 & = &  \exists x\, (\neg U (x) \con M(x)); \medskip\\
  A_2 & = & \forall x \exists y\,  (x \triangleleft y); \medskip\\
  A_3 & = & \forall x \forall y\,  (x \triangleleft y \imp \Box (
            \exists x\, M(x) \imp x
            \triangleleft y)); \medskip\\
  A_4 & = & \forall x \forall y\, (  x \triangleleft y \imp \Box (
            M(x) \equivalence \neg p \con\pDiamond M(y) \con  \neg
            \pDiamond^2 M(y) ); \medskip\\
  A_5 & = & \displaystyle
            \forall x \forall y\, \Box \bigwedge\limits_{t \in T}
            \big( M(x) \con P_t (y) \imp \Box (M(x) \imp P_{t} (y)) \big) ; \medskip\\
  A_6 & = & \displaystyle
            \forall x\, \Box \bigwedge\limits_{t \in T}    (P_t(x) \imp
            \bigwedge\limits_{\mathclap{t' \ne t}} \neg P_{t'}(x)); \medskip\\
  A_7 & = & \displaystyle
            \forall x \forall y\, \Box  \bigwedge\limits_{t \in T}( x \triangleleft y
            \con P_t(x) \imp  \bigvee\limits_{\mathclap{r(t) = l(t')}} P_{t'} (y)) ; \medskip\\
  A_8 & = & \displaystyle
            \forall x \forall y\, \Box \bigwedge\limits_{t \in T} \big(
            M(x) \con P_t (y) \imp \Box (\exists y\, (x
            \triangleleft y \con M(y)) \imp \bigvee\limits_{\mathclap{u(t) =
            d(t')}} P_{t'} (y)) \big); \medskip\\
  A_9 & = & \forall x\, (M(x) \imp \Box \pDiamond P_{t_0} (x)).
\end{array}
$$

Let $A$ be the conjunction of $A_0$ through $A_{9}$.  Observe that $A$
contains only two individual variables.

The relation ${\triangleleft}$ can be thought of as the immediate
successor relation on the domain ${D}_0$ of the world $0$ where $A$ is
being evaluated.  An element $a \in {D}_0$ such that $w \models M(a)$
can be thought of as marking, or labelling, world $w$; thus, we say
that $a$ is a {\em mark} of $w$.  Then, $A_2$ asserts that every
element of $D_0$ has an immediate successor, while $A_3$ asserts that
the immediate successor relation persists throughout the part of the
frame where worlds are marked by elements of $D_0$. Given that, $A_1$
and $A_4$ imply the existence of an infinite sequence
$a_0 \triangleleft a_1 \triangleleft a_2 \triangleleft \ldots$ of
elements of $D_0$ such that every world refuting $p$ is marked, as we
shall see uniquely, by some element of the sequence; they also imply
that the order of the marks of successive, with respect to
$\leqslant$, worlds agrees with the relation $\triangleleft$.  This,
as we shall see, gives us an $\nat \times \nat$ grid whose rows
correspond to the worlds of $\langle \nat, \leqslant \rangle$ and
whose columns correspond to the elements of the sequence
$a_0 \triangleleft a_1 \triangleleft a_2 \triangleleft \ldots\,$.
Building on this, $A_5$ through $A_9$ describe a sought tiling of thus
obtained grid.  (The element of $D_0$ whose existence is asserted by
$A_0$ is not part of the tiling---its presence shall be relied upon in
a subsequent reduction.) 

\begin{lemma}
  \label{lem:tiling-reduction}
  There exists a recurrent tiling of\/ $\nat \times \nat$ satisfying
  \textup{($T_1$)} through \textup{($T_3$)} if, and only if,
  $\langle \nat, \leqslant \rangle \not\models \neg A$.
\end{lemma}

\begin{proof}
  (``if'') Suppose $\frak{M}, w_0 \models A$, for some model
  $\frak{M} = \langle \nat, \leqslant, D, I \rangle$ and some world
  $w_0 \in \nat$.  Since truth of formulas is preserved under taking
  generated submodels,\footnote{The notions of generated subframe and
    generated submodel for predicate modal logics are straightforward
    extensions of the respective notions~\cite[Section 2.1]{BdeRV} for
    propositional modal logics.} we may assume $w_0 = 0$.

  Since $0 \models A_1$, there exists $a_0 \in D_0$ such that
  $0 \not\models U (a_0)$ and $0 \models M(a_0)$.  Since
  $0 \models A_2$, we obtain an infinite sequence
  $a_0, a_1, a_2, \ldots$ of elements of $D_0$ such that
  $a_0 \triangleleft^{I,0} a_1 \triangleleft^{I,0} a_2
  \triangleleft^{I,0} \ldots\,\,$.  Since $0 \models A_3$, we obtain
  that
  $a_0 \triangleleft^{I,w} a_1 \triangleleft^{I,w} a_2
  \triangleleft^{I,w} \ldots\,$, for every $w \in \nat$ such that
  $w \models \exists x\, M(x)$.

  Since $0 \models A_4$, we obtain, for every $w,n \in \nat$,
  $$
  \begin{array}{lcl}
    w \models M(a_n) & \iff  &
    \left\{
         \begin{array}{l}
             w\phantom{'} \not\models p; \\
             w' \models M(a_{n+1}), \mbox{ for some } w' \in R_{\scriptsize\pDiamond}
               (w); \\
             w'' \in R_{\scriptsize\pDiamond}^2 (w) \mbox{ implies } w'' \not\models M(a_{n+1}).
         \end{array}
    \right.
  \end{array}
  \eqno (1)
  $$
  Thus, a mark changes from $a_n$ to $a_{n+1}$ once we pass through a
  world, or an unbroken non-empty sequence of worlds, satisfying $p$
  to a world refuting $p$.

  We now show that a mark remains unchanged until we have reached a
  world satisfying $p$, i.e., that for every $u, u', n \in \nat$,
  $$
  \begin{array}{l}
    \mbox{if $u\models M(a_n)$, $u \not\models p$, $u' \not\models p$,
    and no $v$ with $u \leqslant v \leqslant u'$
    or $u' \leqslant v \leqslant u$} \\
    \mbox{satisfies $v \models p$, then $u' \models M(a_n)$.}
  \end{array}
  \eqno (2)
  $$
  Assume that $u \models M(a_n)$, $u \not\models p$,
  $u' \not\models p$, and that no $v$ with
  $u \leqslant v \leqslant u'$ or $u' \leqslant v \leqslant u$
  satisfies $v \models p$.  Then, by ($1$), there exists
  $w' \in R_{\scriptsize\pDiamond} (u)$ such that
  $w' \models M(a_{n+1})$, and $w'' \not\models M(a_{n+1})$, for every
  $w'' \in R^2_{\scriptsize\pDiamond} (u)$. Let us fix the said $w'$.
  It follows immediately from the assumption that, for every
  $w, n \in \nat$,
  $$
  \begin{array}{lcl}
  w \in R^n_{\scriptsize\pDiamond} (u)
    & \iff
    & w \in R^n_{\scriptsize\pDiamond} (u').
  \end{array}
  $$
  Therefore, $w' \in R_{\scriptsize\pDiamond} (u')$, and
  $w'' \in R^2_{\scriptsize\pDiamond} (u')$ implies
  $w'' \not\models M(a_{n+1})$, for every $w'' \in \nat$.  Since by
  assumption $u' \not\models p$, we obtain, by ($1$), that
  $u' \models M(a_n)$.

  We next show that a mark of every world is unique, i.e. for every
  $w,n \in \nat$ and every $j \in \nat^+$,
  $$
  \begin{array}{lcl}
  w \models M(a_n)
    & \mbox{implies}
    & w \not\models M(a_{n+j}).
  \end{array}
  \eqno (3)
  $$
  Assume $w \models M(a_n)$.  By ($1$), there exists
  $w' \in R^{j+1}_{\scriptsize\pDiamond} (w)$ such that
  $w' \models M(a_{n+j+1})$.  Since $j \geqslant 1$ and
  $R_{\scriptsize\pDiamond}$ is transitive,
  $w' \in R_{\scriptsize\pDiamond}^{2} (w)$.  Therefore, by ($1$),
  $w \not\models M(a_{n+j})$.

  We next show that every element $a_n$ is tiled at every world marked
  by some element $a_m$, i.e. that for every $w, m, n \in \nat$,
  $$
  \begin{array}{lcl}
  w \models M(a_m)
    & \mbox{implies}
    & \mbox{$w \models P_t (a_n)$, for some $t \in T$.}
  \end{array}
  \eqno (4)
  $$

  We proceed by induction on $m$.

  As we have seen, $0 \not\models U(a_0)$, i.e.,
  $0 \models P_t (a_0)$, for some $t \in T$.  Since $0 \models A_7$,
  for every $n \in \nat$, there exists $t \in T$ such that
  $0 \models P_t (a_n)$.  Since $0 \models A_5$, for every
  $w, v, m, n \in \nat$ and every $t \in T$,
  $$
  \begin{array}{lcl}
  \mbox{$w \models M(a_m)$, $v \models M(a_m)$ and $w \models P_t(a_n)$}
     & \mbox{imply}
     & \mbox{$v \models P_t(a_n)$.}
  \end{array}
  \eqno (5)
  $$
  Therefore, ($4$) holds for $m = 0$.

  Assume ($4$) holds for $m \geqslant 0$ and suppose
  $w \models M(a_{m+1})$.  We claim that, then, there exists $w'$ such
  that $w' < w$ and $w' \models M(a_m)$.  To prove the claim we,
  first, observe that there exists $w'$ such that $w' \not\models p$
  and $w \in R_{\scriptsize\pDiamond} (w')$: otherwise, by ($2$),
  $w \models M(a_0)$, in contradiction with ($3$).  Fix the said $w'$.
  We next show that $w'' \in R^2_{\scriptsize\pDiamond} (w')$ implies
  $w'' \not\models M(a_{m+1})$.  Assume
  $w'' \in R^2_{\scriptsize\pDiamond} (w')$.  Then,
  $w'' \in R_{\scriptsize\pDiamond} (w)$.  Since
  $w \models M(a_{m+1})$, by ($1$) and ($2$), $w'' \models M(a_{m+2})$
  and thus, by ($3$), $w'' \not\models M(a_{m+1})$.  Last, since
  $w' \not\models p$, we obtain, by ($1$), $w' \models M(a_m)$,
  thereby proving the claim.

  Now, let $n \in \nat$ be given. By inductive hypothesis, there
  exists $t$ such that $w' \models P_t(a_n)$.  Since $0 \models A_8$,
  this implies that $w \models P_{t'}(a_n)$, for some $t' \in T$.
  Thus, ($4$) is proven.

  In view of ($5$), for every $m \in \nat$, we may pick an arbitrary
  world marked by $a_m \in D_0$ to be part of the sought tiling. For
  definiteness, let, for every $m \in \nat$,
  $$
  w_m = \min \{ w \in \nat : w \models M(a_m) \}.
  $$

  By ($4$), for every $n, m \in \nat$, there exists $t \in T$ such
  that $w_m \models P_t (a_n)$; it follows from $0 \models A_6$ that
  such $t$ is unique.  We can, therefore, define a function
  $f\colon \nat \times \nat \to T$ by
  $$
  \begin{array}{lcl}
    f(n, m) = t
    & \mbox{ whenever } & w_m \models P_t(a_n).
  \end{array}
  $$

  We next show that $f$ satisfies ($T_1$) through ($T_3$).

  Since $0 \models A_7$, the condition ($T_1$) is, evidently,
  satisfied.

  To see that ($T_2$) is satisfied, assume $f(n, m) = t$.  Then,
  $w_m \models P_t (a_n)$, by definition of $f$.  From the definition
  of $w_m$ we know that $w_m \models M(a_m)$.  Since $0 \models A_8$,
  if $v \geqslant w_m$ and $v \models M(a_{m+1})$, then
  $v \models P_{t'} (a_n)$, for some $t'$ with $up(t) = down(t')$.  We
  next show that $w_{m+1} \geqslant w_m$.  Assume otherwise: let
  $w_{m+1} < w_m$.  From the definition of $w_{m+1}$ we know that
  $w_{m+1} \models M(a_{m+1})$.  Therefore, by ($2$) and $(3)$,
  $w_m \in R_{\scriptsize\pDiamond} (w_{m+1})$.  Since
  $w_m \models M(a_m)$, there exists, by ($1$),
  $w' \in R_{\scriptsize\pDiamond}^2 (w_m)$ such that
  $w' \models M(a_{m+2})$.  Since $R_{\scriptsize\pDiamond}$ is
  transitive, $w' \in R_{\scriptsize\pDiamond}^2 (w_{m+1})$, in
  contradiction with the third clause of ($1$).  Thus, we have shown
  that $w_{m+1} \geqslant w_m$.  Hence,
  $w_{m+1} \models P_{t'} (a_n)$, for some $t'$ with
  $up(t) = down(t')$.  Therefore, ($T_2$) is satisfied.

  It remains to show that ($T_3$) is satisfied.  Since
  $0 \models M(a_0)$ and $0 \models A_9$, the set
  $\{ w \in \nat : w \not\models p \mbox{ and } w \models P_{t_0}
  (a_0)\}$ is infinite.  It follows from $0 \models M(a_0)$, ($1$) and
  ($2$) that, for every $w \in \nat$, if $w \not\models p$, then there
  exists $m \in \nat$ such that $w \models M(a_m)$.  Therefore, by
  ($5$), the set
  $\{ w_m : m \in \nat \mbox{ and } w_m \models P_{t_0} (a_0)\}$ is
  infinite.  Hence, ($T_3$) is satisfied.

  Thus, $f$ is a required function.

\begin{figure}
  \centering
  \begin{tikzpicture}[scale=1]

\coordinate (w0)   at (1, 1.0);
\coordinate (w1)   at (1, 2.50);
\coordinate (w2)   at (1, 4.00);
\coordinate (w3)   at (1, 5.50);
\coordinate (w0')  at (1, 1.75);
\coordinate (w1')  at (1, 3.25);
\coordinate (w2')  at (1, 4.75);
\coordinate (w3')  at (1, 6.25);

\coordinate (dots) at (1, 7.00);

\coordinate (a)    at (1, 0.25);

\draw [] (w0)  circle [radius=2.0pt] ;
\draw [] (w1)  circle [radius=2.0pt] ;
\draw [] (w2)  circle [radius=2.0pt] ;
\draw [] (w3)  circle [radius=2.0pt] ;
\draw [] (w0') circle [radius=2.0pt] ;
\draw [] (w1') circle [radius=2.0pt] ;
\draw [] (w2') circle [radius=2.0pt] ;
\draw [] (w3') circle [radius=2.0pt] ;

\draw [] (dots) node {$\vdots$};

\begin{scope}[>=latex]
\draw [->, shorten >= 2.0pt, shorten <= 2.0pt] (w0)  -- (w0');
\draw [->, shorten >= 2.0pt, shorten <= 2.0pt] (w0') -- (w1) ;
\draw [->, shorten >= 2.0pt, shorten <= 2.0pt] (w1)  -- (w1');
\draw [->, shorten >= 2.0pt, shorten <= 2.0pt] (w1') -- (w2) ;
\draw [->, shorten >= 2.0pt, shorten <= 2.0pt] (w2)  -- (w2');
\draw [->, shorten >= 2.0pt, shorten <= 2.0pt] (w2') -- (w3) ;
\draw [->, shorten >= 2.0pt, shorten <= 2.0pt] (w3)  -- (w3');
\end{scope}

\node [left=5pt] at (w0)   {$0$}     ;
\node [left=5pt] at (w0')  {$1$}     ;
\node [left=5pt] at (w1)   {$2$}     ;
\node [left=5pt] at (w1')  {$3$}     ;
\node [left=5pt] at (w2)   {$4$}     ;
\node [left=5pt] at (w2')  {$5$}     ;
\node [left=5pt] at (w3)   {$6$}     ;
\node [left=5pt] at (w3')  {$7$}     ;

\node [right=5pt] at (w0)  {$M(0)$}  ;
\node [right=5pt] at (w1)  {$M(1)$}  ;
\node [right=5pt] at (w2)  {$M(2)$}  ;
\node [right=5pt] at (w3)  {$M(3)$}  ;
\node [right=5pt] at (w0') {$p$}     ;
\node [right=5pt] at (w1') {$p$}     ;
\node [right=5pt] at (w2') {$p$}     ;
\node [right=5pt] at (w3') {$p$}     ;

\node [right = 40pt ] at (w0) {$P_{f(0,0)}(0)$};
\node [right = 95pt ] at (w0) {$P_{f(1,0)}(1)$};
\node [right = 150pt] at (w0) {$P_{f(2,0)}(2)$};
\node [right = 205pt] at (w0) {$P_{f(3,0)}(3)$};
\node [right = 265pt] at (w0) {$\ldots$}       ;

\node [right = 40pt ] at (w1) {$P_{f(0,1)}(0)$};
\node [right = 95pt ] at (w1) {$P_{f(1,1)}(1)$};
\node [right = 150pt] at (w1) {$P_{f(2,1)}(2)$};
\node [right = 205pt] at (w1) {$P_{f(3,1)}(3)$};
\node [right = 265pt] at (w1) {$\ldots$}       ;

\node [right = 40pt ] at (w2) {$P_{f(0,2)}(0)$};
\node [right = 95pt ] at (w2) {$P_{f(1,2)}(1)$};
\node [right = 150pt] at (w2) {$P_{f(2,2)}(2)$};
\node [right = 205pt] at (w2) {$P_{f(3,2)}(3)$};
\node [right = 265pt] at (w2) {$\ldots$}       ;

\node [right = 40pt ] at (w3) {$P_{f(0,3)}(0)$};
\node [right = 95pt ] at (w3) {$P_{f(1,3)}(1)$};
\node [right = 150pt] at (w3) {$P_{f(2,3)}(2)$};
\node [right = 205pt] at (w3) {$P_{f(3,3)}(3)$};
\node [right = 265pt] at (w3) {$\ldots$}       ;

\node [right = 60pt ] at (dots) {$\vdots$};
\node [right = 115pt] at (dots) {$\vdots$};
\node [right = 170pt] at (dots) {$\vdots$};
\node [right = 225pt] at (dots) {$\vdots$};

\node [right =  40pt + 20pt] at (a) {$0$};
\node [right =  65pt + 20pt] at (a) {$\lhd$};
\node [right =  95pt + 20pt] at (a) {$1$};
\node [right = 120pt + 20pt] at (a) {$\lhd$};
\node [right = 150pt + 20pt] at (a) {$2$};
\node [right = 175pt + 20pt] at (a) {$\lhd$};
\node [right = 205pt + 20pt] at (a) {$3$};
\node [right = 245pt + 20pt] at (a) {$\ldots$}       ;

\end{tikzpicture}

\caption{Model $\frak{M}_0$}
  \label{fig1}
\end{figure}

(``only if'') Suppose $f$ is a function satisfying ($T_1$) through
($T_3$).  We obtain a model based on $\langle \nat, \leqslant \rangle$
satisfying $A$.

  Let $\mathcal{D} = \nat \cup \{-1\}$ and $D(w) = \mathcal{D}$, for
  every $w \in \nat$.

  Let $\frak{M}_0 = \langle \nat, \leqslant, D, I \rangle$ be a model
  such that, for every $w \in \nat$ and every $a,b\in {\cal D}$,
  $$
  \begin{array}{lcl}
    \frak{M}_0,w \models a\triangleleft b
    & \leftrightharpoons &
                           \mbox{$w$ is even and $b=a+1$;} \smallskip \\
    \frak{M}_0,w \models p
    & \leftrightharpoons &
                           \mbox{$w$ is odd;} \smallskip \\
    \frak{M}_0,w \models M(a)
    & \leftrightharpoons &
                           \mbox{$w=2a$;} \smallskip \\
    \frak{M}_0,w \models P_t(a)
    & \leftrightharpoons &
                           \mbox{for some $m\in\nat$, both $w=2m$ and $f(a,m) = t$.}
  \end{array}
  $$

  It is straightforward to check that $\frak{M}_0, 0 \models A$, so we
  leave this to the reader.
\end{proof}

Thus, in the proof of the ``if'' part of
Lemma~\ref{lem:tiling-reduction}, we obtained a grid for the tiling by
treating the worlds of model $\frak{M}$ as rows and elements
$a_0, a_1, a_2, \ldots$ of the domain $D_0$ of the world $0$
satisfying $A$ as columns.

\subsection{Elimination of the binary predicate letter}
\label{sec:elim-binary}

We next eliminate, following ideas of Kripke's~\cite{Kripke62}, the
binary predicate letter $\triangleleft$ of formula $A$, without
increasing the number of individual variables in the resultant
formula.

From now on, we assume, for ease of notation, that $A$ contains
monadic predicate letters $P_0, \ldots, P_s$---rather than $P_t$, for
each $t \in \{t_0, \ldots, t_s \}$---to refer to the tile types.

Recall that Kripke's construction~\cite{Kripke62} transforms a model
$\frak{M}$ satisfying, at world $w$, a formula containing a binary
predicate letter, and no modal connectives, so that, for every pair of
elements of the domain of $w$, a fresh world accessible from $w$ is
introduced to $\frak{M}$.  This construction cannot be applied here in
a straightforward manner, for two reasons.

First, since we are working with the frame
$\langle \nat, \leqslant \rangle$, we may not introduce fresh worlds
to a model satisfying $A$; we, rather, have to use the worlds from
$\nat$ to simulate $\triangleleft$.  Second, since $\triangleleft$
occurs within the scope of the modal connective in $A$, we need to
simulate the interpretation of $\triangleleft$ not just at the world
satisfying $A$, but at every world accessible from it.

We resolve these difficulties by working with the model $\frak{M}_0$
defined in the ``only if'' part of the proof of
Lemma~\ref{lem:tiling-reduction}, rather than with an arbitrary model
satisfying $A$, and relying on $\frak{M}_0$ being based on a frame
with a constant domain and on the interpretation of $\triangleleft$
being identical at every world of $\frak{M}_0$.

Let $P_{s+1}$ and $P_{s+2}$ be monadic predicate letters distinct from
$M, P_0, \ldots, P_s$ and from each other,
and let $\cdot'$ be the function substituting
$$
\begin{array}{lcl}
\pDiamond (P_{s+1} (x) \con P_{s+2} (y))
  & \mbox{for}
  & x \triangleleft y.
\end{array}
$$

\begin{lemma}
  \label{lem:Kripke}
  There exists a recurrent tiling of\/ $\nat \times \nat$\/ satisfying
  \textup{($T_1$)} through \textup{($T_3$)} if, and only if,
  $\langle \nat, \leqslant \rangle \not\models \neg A'$.
\end{lemma}

\begin{proof} (``if'') Suppose $\frak{M}, w_0 \models A'$, for some
  model $\frak{M} = \langle \nat, \leqslant, D, I \rangle$ and some
  world $w_0$, which can be assumed to be $0$.

  The argument is essentially the same as in the proof of the ``if''
  part of Lemma~\ref{lem:tiling-reduction}.  The only,
  inconsequential, difference is that
  $\pDiamond (P_{s+1} (x) \con P_{s+2} (y))$ now plays the role of
  $x \triangleleft y$: for every $w \in \nat$, the relation
  $I(w, \triangleleft) \subseteq D_w \times D_w$ is replaced by the
  relation
  $$
  \{ \langle a, b \rangle \in D_w \times D_w : \frak{M}, w \models
  \pDiamond (P_{s+1} (a) \con P_{s+2} (b))\}.
  $$
  Since $\frak{M}, 0 \models A'$, the two relations are
  indistinguishable, for every $w \in \nat$, with respect to the
  properties we rely on in the proof.

  (``only if'') Suppose $f$ is a function satisfying ($T_1$) through
  ($T_3$).  Let $\frak{M}_0 = \langle \nat, \leqslant, D, I \rangle$
  be the model defined in the ``only if'' part of the proof of
  Lemma~\ref{lem:tiling-reduction}.  As we have seen,
  $\frak{M}_0, 0 \models A$.  We use $\frak{M}_0$ to obtain a model
  satisfying $A'$.

  Let $\alpha$ be the infinite sequence
  $$
  0, \, \, 0, 1, \, \, 0, 1, 2, \, \, 0, 1, 2, 3,  \, \, 0, 1, 2, 3,
  4,\, \,  \ldots
  $$
  and let $\alpha_k$, for each $k \in \nat$, be the $k$th element of
  $\alpha$.

  Let $\frak{M}'_0 = \langle \nat, \leqslant, D, I' \rangle$ be a
  model such that, for every $w, c \in \nat$,
  $$
  \begin{array}{lcl}
    \frak{M}'_0, w \models P_{s+1} (c)
    & \leftrightharpoons &
    \mbox{for some $m\in\nat$, both $w=2m$ and $c=\alpha_m$;}
    \medskip \\
    \frak{M}'_0, w \models P_{s+2} (c)
    & \leftrightharpoons &
    \mbox{for some $m\in\nat$, both $w=2m$ and $c=\alpha_m+1$,}
  \end{array}
  $$
  and for every $w \in \nat$ and every
  $S \in \{P_0, \ldots, P_s, M, p\}$,
  $$
  \begin{array}{lcl}
  I'(w,S) & = & I(w,S).
  \end{array}
  $$

  We show that $\frak{M}'_0, 0 \models A'$.

  Since $\frak{M}_0, 0 \models A$, it suffices to prove that, for
  every $m \in \nat$ and every $a, b \in \mathcal{D}$,
  $$
  \begin{array}{lcl}
  \frak{M}_0, 2m \models a \triangleleft b
    & \iff
    & \frak{M}'_0, 2m \models \pDiamond (P_{s+1} (a) \con P_{s+2}
  (b)).
  \end{array}
  $$

  Assume $\frak{M}_0, 2m \models a \triangleleft b$.  Then,
  $b = a + 1$, by definition of $\frak{M}_0$.  Choose $k \in \nat$ so
  that $k > m$ and $\alpha_k = a$; by definition of $\alpha$, such a
  number $k$ certainly exists.  By definition of $\frak{M}'_0$, both
  $\frak{M}'_0, 2k \not\models p$ and
  $\frak{M}'_0, 2k \models P_{s+1} (a) \con P_{s+2} (b)$.  By the same
  definition, $\frak{M}'_0, 2k - 1 \models p$.  Hence,
  $\frak{M}'_0, 2m \models \pDiamond (P_{s+1} (a) \con P_{s+2} (b))$.

  Conversely, assume
  $\frak{M}'_0, 2m \models \pDiamond (P_{s+1} (a) \con P_{s+2} (b))$.
  Then, for some $v > 2m$, both $\frak{M}'_0, v \not\models p$ and
  $\frak{M}'_0, v \models P_{s+1} (a) \con P_{s+2} (b)$.  By
  definition of $\frak{M}'_0$, we have
  $\frak{M}_0, v \not\models p$; hence $v = 2k$, for some $k > m$.
  Also by definition of $\frak{M}'_0$, both $a = \alpha_k$ and
  $b = \alpha_k + 1$; hence, $b = a + 1$.  Therefore,
  $\frak{M}_0, 2m \models a \triangleleft b$, by definition of
  $\frak{M}_0$.~
\end{proof}

\subsection{Elimination of monadic predicate letters}
\label{sec:elim-monadic}

We lastly simulate the occurrences of letters
$p, M, P_0, \ldots, P_{s+2}$ in $A'$ with one monadic and one
proposition letter, without increasing the number of individual
variables in the resultant formula.

Let $P$ be a monadic letter distinct from
$M, P_0, \ldots, P_{s+2}$, and let $q$ be a proposition letter
distinct from $p$.

For a formula $\vp$ in the language containing $P$ and $q$, define
$$
\begin{array}{c}
  \PDiamond\varphi = \Diamond(\forall x\, P(x) \con \Diamond( \neg \forall x\, P(x) \con\varphi));
  \medskip\\
  \PDiamond^0 \varphi = \varphi; \quad \PDiamond^{n+1} \varphi =
  \PDiamond\PDiamond^n\varphi, \mbox{ for every } n \in \nat.
\end{array}
$$

Define, for every $n \in \{0, \ldots, s + 2\}$,
$$
\begin{array}{rcl}
  \beta_n (x) & = & \exists y\, \big( \PDiamond^{s+4} (q \con P(y)) \con \neg
                    \PDiamond^{s+5} (q \con P(y))\, \con \\
              & & \phantom{\mu \con P(x) \con \exists y\, [} \PDiamond (\PDiamond^{n+1} (q \con P(y)) \con \neg
                  \PDiamond^{n+2} (q \con P(y)) \con P(x)) \big);\medskip\\
  \beta_n (y) & = & \exists x\, \big( \PDiamond^{s+4} (q \con P(x)) \con \neg
                    \PDiamond^{s+5} (q \con P(x))\, \con \\
              & & \phantom{\mu \con P(y) \con \exists x\, [} \PDiamond (\PDiamond^{n+1} (q \con P(x)) \con \neg
                  \PDiamond^{n+2} (q \con P(x)) \con P(y)) \big).\\
\end{array}
$$

Let $\cdot^\ast$ be the function replacing
\begin{itemize}
\item $P_n(x)$ with $\beta_n (x)$, for every
  $n \in \{0, \ldots, s + 2\}$;
\item $P_n(y)$ with $\beta_n (y)$, for every
  $n \in \{0, \ldots, s + 2\}$;
\item $M(x)$ with $q \con P(x)$;
\item $M(y)$ with $q \con P(y)$.
\end{itemize}

Let $A_i^\ast$, for each $i$ with $0 \leqslant i \leqslant 8$ and
$i \ne 4$, be the result of applying the function $\cdot^\ast$ to
$A'_i$.  Also, let
$$
\begin{array}{lcl}
  A_4^{\ast} & = & \forall x \forall y\, \big( \PDiamond (\beta_{s+1}
                   (x) \con \beta_{s+2} (y) )\, \imp \smallskip\\
             & & \phantom{\forall x \forall y\, ( \PDiamond (\beta_{s+1}}\Box (q \con P(x) \equivalence
                 \neg \forall x\, P(x) \con \PDiamond^{s+4} ( q \con P(y)) \con \neg
                 \PDiamond^{s+5} (q \con P(y))) \big),
\end{array}
$$
and
$$
\begin{array}{lcl}
A^\ast_{9} & = & \forall x\, (q \con P(x) \imp \Box \PDiamond \beta_{0} (x)).
\end{array}
$$

Lastly, let $A^\ast$ be the conjunction of $A_0^\ast$ through
$A_9^\ast$.  Observe that $A^\ast$ contains only two individual
variables, a monadic letter $P$ and a proposition letter $q$.

We shall show that $A^\ast$ is satisfiable if, and only if, there
exists a recurrent tiling satisfying ($T_1$) through ($T_3$).

\begin{figure}
  \centering
  \begin{tikzpicture}[scale=1]

\coordinate (w0)   at (1, -0.5);
\coordinate (ws)   at (1, 1.5);
\coordinate (ws`)  at (1, 2.75);
\coordinate (vsn2) at (1, 3.50);
\coordinate (vsn2`) at (1, 4.25);
\coordinate (vsn1) at (1, 5.0);
\coordinate (vsn1`) at (1, 5.75);
\coordinate (vsn)  at (1, 6.50);
\coordinate (vs1)  at (1, 7.50);
\coordinate (vs1`)  at (1, 8.25);
\coordinate (vs0)  at (1, 9.0);
\coordinate (vs0`)  at (1, 9.75);
\coordinate (ws')  at (1, 11.00);

\coordinate (dots)  at (1, 11.8);
\coordinate (dots0) at (1, 0.6);
\coordinate (dotsv) at (1, 7.1);

\draw [] (w0)    circle [radius=2.0pt] ;
\draw [] (ws)    circle [radius=2.0pt] ;
\draw [] (ws')   circle [radius=2.0pt] ;
\draw [] (vs0)   circle [radius=2.0pt] ;
\draw [] (vs1)   circle [radius=2.0pt] ;
\draw [] (vsn)   circle [radius=2.0pt] ;
\draw [] (vsn1)  circle [radius=2.0pt] ;
\draw [] (vsn2)  circle [radius=2.0pt] ;

\draw [] (ws`)   circle [radius=2.0pt] ;
\draw [] (vs0`)  circle [radius=2.0pt] ;
\draw [] (vs1`)  circle [radius=2.0pt] ;
\draw [] (vsn1`) circle [radius=2.0pt] ;
\draw [] (vsn2`) circle [radius=2.0pt] ;

\draw (dots)  node {$\vdots$};
\draw (dots0) node {$\vdots$};
\draw (dotsv) node {$\vdots$};

\begin{scope}[>=latex]
\draw [->, shorten >= 2pt, shorten <= 2pt] (ws)    -- (ws`)  ;
\draw [->, shorten >= 2pt, shorten <= 2pt] (ws`)   -- (vsn2) ;
\draw [->, shorten >= 2pt, shorten <= 2pt] (vsn2)  -- (vsn2`);
\draw [->, shorten >= 2pt, shorten <= 2pt] (vsn2`) -- (vsn1) ;
\draw [->, shorten >= 2pt, shorten <= 2pt] (vsn1)  -- (vsn1`);
\draw [->, shorten >= 2pt, shorten <= 2pt] (vsn1`) -- (vsn)  ;
\draw [->, shorten >= 2pt, shorten <= 2pt] (vs1)   -- (vs1`) ;
\draw [->, shorten >= 2pt, shorten <= 2pt] (vs1`)  -- (vs0)  ;
\draw [->, shorten >= 2pt, shorten <= 2pt] (vs0)   -- (vs0`) ;
\draw [->, shorten >= 2pt, shorten <= 2pt] (vs0`)  -- (ws')  ;
\end{scope}

\node [left = 5pt] at (w0)    {$w_0$}              ;
\node [left = 5pt] at (ws)    {$w_m$}              ;
\node [left = 5pt] at (ws')   {$w_{m+1}$}          ;
\node [left = 5pt] at (vs0)   {$v_m^0$}            ;
\node [left = 5pt] at (vs1)   {$v_m^1$}            ;
\node [left = 5pt] at (vsn)   {$v_m^s$}            ;
\node [left = 5pt] at (vsn1)  {$v_m^{s+1}$}        ;
\node [left = 5pt] at (vsn2)  {$v_m^{s+2}$}        ;
\node [left = 5pt] at (ws`)   {$\bar w_m$}         ;
\node [left = 5pt] at (vs0`)  {$\bar v_m^0$}       ;
\node [left = 5pt] at (vs1`)  {$\bar v_m^1$}       ;
\node [left = 5pt] at (vsn1`) {$\bar v_m^{s+1}$}   ;
\node [left = 5pt] at (vsn2`) {$\bar v_m^{s+2}$}   ;

\node [right=5pt]  at (w0)    {$q$}                ;
\node [right=5pt]  at (ws)    {$q$}                ;
\node [right=5pt]  at (ws')   {$q$}                ;
\node [right=20pt] at (w0)    {$P(0)$}             ;
\node [right=20pt] at (ws)    {$P(m)$}             ;
\node [right=20pt] at (ws')   {$P(m+1)$}           ;
\node [right=5pt]  at (ws`)   {$\forall x\, P(x)$} ;
\node [right=5pt]  at (vs0`)  {$\forall x\, P(x)$} ;
\node [right=5pt]  at (vs1`)  {$\forall x\, P(x)$} ;
\node [right=5pt]  at (vsn1`) {$\forall x\, P(x)$} ;
\node [right=5pt]  at (vsn2`) {$\forall x\, P(x)$} ;

\node [right = 40pt  + 15pt + 25pt] at (w0) {$\beta_{f(0,\,\,0)}(0)$};
\node [right =  95pt + 15pt + 35pt] at (w0) {$\beta_{f(1,\,\,0)}(1)$};
\node [right = 150pt + 15pt + 45pt] at (w0) {$\beta_{f(2,\,\,0)}(2)$};
\node [right = 205pt + 15pt + 55pt] at (w0) {$\beta_{f(3,\,\,0)}(3)$};
\node [right = 275pt + 15pt + 65pt] at (w0) {$\ldots$}               ;

\node [right =  40pt + 15pt + 25pt] at (ws) {$\beta_{f(0,m)}(0)$}    ;
\node [right =  95pt + 15pt + 35pt] at (ws) {$\beta_{f(1,m)}(1)$}    ;
\node [right = 150pt + 15pt + 45pt] at (ws) {$\beta_{f(2,m)}(2)$}    ;
\node [right = 205pt + 15pt + 55pt] at (ws) {$\beta_{f(3,m)}(3)$}    ;
\node [right = 275pt + 15pt + 65pt] at (ws) {$\ldots$}               ;

\node [right =  40pt + 25pt + 25pt] at (ws') {$\beta_{f(0,m+1)}(0)$} ;
\node [right =  95pt + 25pt + 35pt] at (ws') {$\beta_{f(1,m+1)}(1)$} ;
\node [right = 150pt + 25pt + 45pt] at (ws') {$\beta_{f(2,m+1)}(2)$} ;
\node [right = 205pt + 25pt + 55pt] at (ws') {$\beta_{f(3,m+1)}(3)$} ;
\node [right = 275pt + 15pt + 65pt] at (ws') {$\ldots$}              ;

\node [right = 20pt] at (vs0)  {$P(n) \iff f(n,m) = t_0$};
\node [right = 20pt] at (vs1)  {$P(n) \iff f(n,m) = t_1$};
\node [right = 20pt] at (vsn)  {$P(n) \iff f(n,m) = t_s$};
\node [right = 20pt] at (vsn1) {$P(\alpha_m)$};
\node [right = 20pt] at (vsn2) {$P(\alpha_m + 1)$};


\end{tikzpicture}

\caption{Model $\frak{M}_0^\ast$}
  \label{fig2}
\end{figure}

To obtain a model satisfying $A^\ast$, we ``stretch out'' the model
$\frak{M}'_0$ defined in the ``only if'' part of the proof of
Lemma~\ref{lem:Kripke} to include ``additional'' worlds whose sole
purpose is to simulate the interpretation of letters
$P_0, \ldots, P_{s+2}$ at worlds of $\frak{M}'_0$.  We ``insert''
$s+3$ worlds between worlds $m$ and $m + 1$ to simulate the
interpretation of letters $P_0, \ldots, P_{s+2}$ at $m$.  The
interpretation of $P_n$, for each $n \in \{0, \ldots, s+2\}$, at $m$
is simulated by the interpretation of letter $P$ at a newly inserted
world ``$n$ steps away from'' $m+1$.  To be able to step through the
newly defined model, we also ``insert'' extra worlds satisfying
$\forall x\, P(x)$; these play the same role the worlds satisfying $p$
played in $\frak{M}'_0$. The proposition letter $q$ marks off the
``old'' worlds from $\frak{M}'_0$. The resultant model is depicted in
Figure~\ref{fig2}, where $\beta_{f(a,b)}(x)$ stands for $\beta_n (x)$,
where $n$ is such that $f(a,b)=t_n$.

\begin{lemma}
  \label{lem:monadic}
  There exists a recurrent tiling of\/ $\nat \times \nat$\/ satisfying
  \textup{($T_1$)} through \textup{($T_3$)} if, and only if,
  $\langle \nat, \leqslant \rangle \not\models \neg A^\ast$.
\end{lemma}

\begin{proof}
  (``if'')   Suppose $\frak{M}, w_0 \models A^\ast$, for some
  model $\frak{M} = \langle \nat, \leqslant, D, I \rangle$ and some
  world $w_0$, which can be assumed to be $0$.

  The argument is essentially the same as in the proof of the ``if''
  part of Lemma~\ref{lem:Kripke}, the only difference being that we
  use
  \begin{itemize}
  \item $\beta_n(x)$ instead of $P_n(x)$;
  \item $\beta_n(y)$ instead of $P_n(y)$;
  \item $q \con P(x)$ instead of $M(x)$;
  \item $q \con P(y)$ instead of $M(y)$.
  \end{itemize}
  
  (``only if'') Suppose $f$ is a function satisfying ($T_1$) through
  ($T_3$).  Let $\frak{M}'_0 = \langle \nat, \leqslant, D, I' \rangle$
  be the model defined in the ``only if'' part of the proof of
  Lemma~\ref{lem:Kripke}. As we have seen,
  $\frak{M}'_0, 0 \models A'$.  We use $\frak{M}'_0$ to obtain a model
  satisfying $A^{\ast}$.

  We think of the worlds from $\nat$ as being labeled, in the
  ascending order,
  $$
  \begin{array}{l}
  w_0^{\phantom{i}}, \bar{w}_0^{\phantom{i}}, v^{s+2}_0,
  \bar{v}^{s+2}_0, \ldots, v^0_0, \bar{v}^0_0, \\
  w_1^{\phantom{i}}, \bar{w}_1^{\phantom{i}}, v^{s+2}_1,
  \bar{v}^{s+2}_1, \ldots, v^0_1, \bar{v}^0_1, \\
    w_2^{\phantom{i}}, \bar{w}_2^{\phantom{i}}, \ldots\, ,
  \end{array}
  $$
  i.e., we put $w_0 = 0$, $\bar{w}_0^{\phantom{i}} = 1$,
  $v^{s+2}_0 = 2$, etc.

  Let $\frak{M}_0^\ast = \langle \nat, \leqslant, D, I^\ast \rangle$
  be a model such that, for every $u \in
  \nat$
  $$
  \begin{array}{lcll}
    \frak{M}_0^\ast, u \models q
      & \leftrightharpoons
      & \mbox{$u = w_m$, for some $m \in \nat$,}
  \end{array}
  $$
  and for every $u \in \nat$ and every
  $a \in \mathcal{D}$, the relation $\frak{M}_0^\ast, u \models P(a)$
  holds if, and only if, one of the following conditions is satisfied:
  \begin{itemize}
  \item \mbox{$u = w_m$ and $\frak{M}'_0, 2m \models M(a)$, for some
      $m \in \nat$};
  \item \mbox{$u = v^n_m$ and $\frak{M}'_0, 2m \models P_n(a)$, for
      some $m \in \nat$ and some $n \in \{0, \ldots, s+2\}$;}
  \item \mbox{$u = \bar{w}_m$, for some $m \in \nat$;}
  \item \mbox{$u = \bar{v}^n_m$, for some $m \in \nat$
                and some $n \in \{0, \ldots, s+2\}$.}
  \end{itemize}

  Thus, by definition of $\frak{M}^\ast_0$,
  $$
  \begin{array}{lcl}
    \frak{M}^\ast_0, w_m \models q \con P(a) & \iff & a = m.
  \end{array}
  \eqno (6)
  $$

  We now prove that $\frak{M}_0^\ast, w_0 \models A^{\ast}$.

  First, we show that
  $$
  \begin{array}{lcl}
  \frak{M}^\ast_0, u \models \forall x\, P(x)
  & \iff &
  u \in \{ \bar{w}_m : m \in \nat \} \cup \{
  \bar{v}^n_m : m \in \nat, 0 \leqslant n \leqslant s
  + 2 \}.
  \end{array}
  \eqno (7)
  $$
  The right-to-left implication is immediate from the definition of
  $\frak{M}^\ast_0$.

  For the converse, assume
  $u \notin \{ \bar{w}_m : m \in \nat \} \cup \{ \bar{v}^n_m
  : m \in \nat, 0 \leqslant n \leqslant s + 2 \}$.

  We have four cases to consider.

  Case $u = w_m$: The definition of $\frak{M}^\ast_0$ implies that
  $\frak{M}^\ast_0, w_m \not\models P(c)$, for every
  $c \in \mathcal{D} - \{m\}$.  Since
  $\mathcal{D} - \{m\} \ne \varnothing$, we obtain
  $\frak{M}^\ast_0, w_m \not\models \forall x\, P(x)$.

  Case $u = v^{s+1}_m$: The definition of $\frak{M}^\ast_0$ implies
  that $\frak{M}^\ast_0, v^{s+1}_m \models P(a)$ if, and only if,
  $\frak{M}'_0, 2m \models P_{s+1} (a)$, which by definition of
  $\frak{M}'_0$, holds if, and only if,
  $\frak{M}_0, 0 \models a \triangleleft b$ and $\alpha_m = a$.
  Therefore, $\frak{M}^\ast_0, v^{s+1}_m \not\models P(c)$, for every
  $c \in \mathcal{D} -\{\alpha_m\}$. Since
  $\mathcal{D} - \{\alpha_m\} \ne \varnothing$, we obtain
  $\frak{M}^\ast_0, v^{s+1}_m \not\models \forall x\, P(x)$.

  Case $u = v^{s+2}_m$: The definition of $\frak{M}^\ast_0$ implies
  that $\frak{M}^\ast_0, v^{s+2}_m \models P(a)$ if, and only if,
  $\frak{M}'_0, 2m \models P_{s+2} (a)$, which by definition of
  $\frak{M}'_0$, holds if, and only if, $a = \alpha_m +1$.  Therefore,
  $\frak{M}^\ast_0, v^{s+2}_m \not\models P(c)$, for every
  $c \in \mathcal{D} -\{\alpha_m + 1\}$. Since
  $\mathcal{D} - \{\alpha_m + 1\} \ne \varnothing$, we obtain
  $\frak{M}^\ast_0, v^{s+2}_m \not\models \forall x\, P(x)$.

  Case $u = v^{n}_m$, where $n \in \{0, \ldots, s\}$: By definitions
  of $\frak{M}^\ast_0$ and $\frak{M}'_0$,
  $$
  \begin{array}{lclcl}
  \frak{M}^\ast_0, v^{n}_m \models P(a)
  & \iff &
  \frak{M}'_0, 2m \models P_{n} (a)
  & \iff &
  \frak{M}_0, 2m \models P_{n} (a).
  \end{array}
  $$
  The definition of $\frak{M}_0$ implies that
  $\frak{M}_0, 2m \not\models P_n(-1)$.  Therefore,
  $\frak{M}^\ast_0, v^{n}_m \not\models P(-1)$; hence,
  $\frak{M}^\ast_0, v^{n}_m \not\models \forall x\, P(x)$.

  Thus, $\frak{M}^\ast_0, u \not\models \forall x\, P(x)$, and so
  ($7$) is proven.

  Next, we show that, for every $m \in \nat$, every
  $n \in \{0, \ldots, s+2\}$ and every $a \in \mathcal{D}$,
  $$
  \begin{array}{lcl}
     \frak{M}_0^\ast, w_m\models\beta_n(a) & \iff & \frak{M}'_0, 2m \models P_n(a).
  \end{array}
  \eqno (8)
  $$

  First, define a binary relation ${R}_{\scriptsize\PDiamond}$ on
  $\nat$ by
  $$
  \begin{array}{lcl}
  w R_{\scriptsize\PDiamond} v
    & \leftrightharpoons
    & v \not\models
  \forall x\, P(x) \mbox{ and, for some $u \in \nat$, both } w
  \leqslant u \leqslant v \mbox{ and } u \models \forall x\, P(x).
  \end{array}
  $$

  Now, assume $\frak{M}'_0, 2m \models P_n(a)$.

  By~($6$), $\frak{M}_0^\ast, w_{m+1} \models q \con P(m+1)$.
  By~($7$) and the definition of $\frak{M}_0^\ast$,
  \begin{itemize}
  \item $w_{m+1} \in R^{s+4}_{\scriptsize\PDiamond}(w_m) - R^{s+5}_{\scriptsize\PDiamond}(w_m)$;
  \item $w_m < v^n_m$;
  \item $w_{m+1} \in R^{n+1}_{\scriptsize\PDiamond}(v^n_m) - R^{n+2}_{\scriptsize\PDiamond}(v^n_m)$;
  \item $\frak{M}_0^\ast, v^n_m \models P(a)$.
  \end{itemize}
  Therefore, $\frak{M}_0^\ast, w_m \models \beta_n (a)$.

  Conversely, assume $\frak{M}_0^\ast, w_m \models \beta_n(a)$.

  Then,
  $$
  \frak{M}_0^\ast, w_m \models \exists y\, (\PDiamond^{s+4} (q \con
  P(y)) \con \neg \PDiamond^{s+5} (q \con P(y))).$$ Hence, there exist
  $u \in \nat$ and $b \in \mathcal{D}$ such that
  $$
  \frak{M}_0^\ast, u \models q \con P(b) \mbox{ and }
  u \in R^{s+4}_{\scriptsize\PDiamond}(w_m) -
  R^{s+5}_{\scriptsize\PDiamond}(w_m).
  $$
  By definition of $\frak{M}^\ast_0$ and by ($6$), the only choices
  for $u$ and $b$ are, respectively, $w_{m+1}$ and $m+1$.  Hence,
  $$
  \frak{M}_0^\ast, w_m \models \PDiamond (\PDiamond^{n+1} (q \con
  P(m+1)) \con \neg \PDiamond^{n+2} (q \con P(m+1)) \con P(a)).
  $$
  Thus, by definition of $\frak{M}_0^\ast$, we obtain
  $\frak{M}_0^\ast, v^n_m \models P(a)$ and, hence,
  $\frak{M}'_0, 2m \models P_n (a)$.  Thus, ($8$) is proven.

  From ($6$), ($7$) and ($8$), we obtain
  $\frak{M}_0^\ast, w_0 \models A_i^{\ast}$, for each $i$ with
  $0 \leqslant i \leqslant 8$ and $i \ne 4$.  Furthermore, based on
  ($6$), ($7$) and ($8$), it is straightforward to check that
  $\frak{M}_0^\ast, w_0 \models A^\ast_4$ and
  $\frak{M}_0^\ast, w_0 \models A^\ast_9$.

  Thus, $\frak{M}_0^\ast, w_0 \models A^{\ast}$.~
\end{proof}

From Lemma~\ref{lem:monadic} we immediately obtain the following
result:

\begin{theorem}
  \label{thr:nat-ref-sat}
  Satisfiability for\/ $\mathbf{L}(\nat, \leqslant)$ is\/
  $\Sigma^1_1$-hard in languages with two individual variables, one
  monadic predicate letter and one proposition letter.
\end{theorem}

\section{Logics of discrete linear orders}
\label{sec:discrete}

We now generalise Theorem~\ref{thr:nat-ref-sat} to logics of descrete
linear orders other than $\langle \nat, \leqslant \rangle$.

We first consider logics of frames based on $\nat$.  Let $R$ be a
binary relation on $\nat$ between $<$ and $\leqslant$.  Define
$$
\begin{array}{lcl}
\Box^{+} \vp & = & \vp \con \Box \vp.
\end{array}
$$
Then, the relation $\leqslant$ is the reflexive closure of $R$ and,
hence, it is the accessibility relation associated with the operator
$\Box^{+}$: for every model
$\frak{M} = \langle \nat, R, D, I \rangle$, every $w \in \nat$ and
every assignment $g$,
$$
\begin{array}{lcl}
  \frak{M},w\models^g \Box^+\varphi & \iff & \mbox{$\frak{M},w'\models^g \varphi$, for every $w'\in \nat$ such that $w\leqslant w'$.}
\end{array}
$$

Let $A^+$ to be the formula obtained from $A^\ast$ by replacing every
occurrence of $\Box$ with an occurrence of $\Box^+$.  The noted
correspondence between $\Box^+$ and $\leqslant$, as well as their
connection with, respectively, $\Box$ and $R$, give us the following
analogue of Lemma~\ref{lem:monadic}:

\begin{lemma}
  \label{lem:monadic-gen}
  There exists a recurrent tiling of\/ $\nat \times \nat$\/ satisfying
  \textup{($T_1$)} through \textup{($T_3$)} if, and only if,
  $\langle \nat, R \rangle \not\models \neg A^+$.
\end{lemma}

From Lemma~\ref{lem:monadic-gen}, we obtain the analogue of
Theorem~\ref{thr:nat-ref-sat} for $\mathbf{L}(\nat, R)$.  Moreover,
since none of the arguments made so far depend on the assumption of
properly expanding domains, we obtain the following generalisation of
Theorem~\ref{thr:nat-ref-sat}:

\begin{theorem}
  \label{thr:nat}
  Let $R$ be a binary relation on $\nat$ between $<$ and $\leqslant$,
  and let $L$ be a logic such that
  $\mathbf{L}(\nat, R) \subseteq L \subseteq \mathbf{L}_{c} (\nat,
  R)$. Then, satisfiability for\/ $L$ is\/ $\Sigma^1_1$-hard in
  languages with two individual variables, one monadic predicate
  letter and one proposition letter.
\end{theorem}

\begin{figure}
  \centering
  \begin{tikzpicture}[scale=1]

\coordinate (w0)     at (1, 1.00);
\coordinate (wn`)    at (1, 2.00);
\coordinate (wn)     at (1, 2.75);
\coordinate (wn')    at (1, 3.50);
\coordinate (wn'')   at (1, 4.25);
\coordinate (wdots)  at (1, 1.60);
\coordinate (wdots') at (1, 4.85);

\coordinate (v0)     at (5, 1.00);
\coordinate (vn`)    at (5, 2.00);
\coordinate (vn)     at (5, 2.75);
\coordinate (vn')    at (5, 3.50);
\coordinate (vn'')   at (5, 4.25);
\coordinate (vdots)  at (5, 1.60);
\coordinate (vdots') at (5, 4.85);

\draw [fill] (w0)     circle [radius=2.0pt] ;
\draw [fill] (wn`)    circle [radius=2.0pt] ;
\draw []     (wn)     circle [radius=2.0pt] ;
\draw []     (wn')    circle [radius=2.0pt] ;
\draw []     (wn'')   circle [radius=2.0pt] ;
\draw []     (wdots)  node   {$\vdots$}     ;
\draw []     (wdots') node   {$\vdots$}     ;

\draw []     (v0)     circle [radius=2.0pt] ;
\draw []     (vn`)    circle [radius=2.0pt] ;
\draw [fill] (vn)     circle [radius=2.0pt] ;
\draw [fill] (vn')    circle [radius=2.0pt] ;
\draw [fill] (vn'')   circle [radius=2.0pt] ;
\draw []     (vdots)  node   {$\vdots$}     ;
\draw []     (vdots') node   {$\vdots$}     ;

\begin{scope}[>=latex]
\draw [->, shorten >= 2.0pt, shorten <= 2.0pt] (wn`)  -- (wn)  ;
\draw [->, shorten >= 2.0pt, shorten <= 2.0pt] (wn)   -- (wn') ;
\draw [->, shorten >= 2.0pt, shorten <= 2.0pt] (wn')  -- (wn'');
\draw [->, shorten >= 2.0pt, shorten <= 2.0pt] (vn`)  -- (vn)  ;
\draw [->, shorten >= 2.0pt, shorten <= 2.0pt] (vn)   -- (vn') ;
\draw [->, shorten >= 2.0pt, shorten <= 2.0pt] (vn')  -- (vn'');
\end{scope}

\node [right=2pt] at (w0)   {$0$}  ;
\node [right=2pt] at (wn`)  {$n-1$};
\node [right=2pt] at (wn)   {$n$}  ;
\node [right=2pt] at (wn')  {$n+1$};
\node [right=2pt] at (wn'') {$n+2$};

\node [right=2pt] at (v0)   {$0$}  ;
\node [right=2pt] at (vn`)  {$n-1$};
\node [right=2pt] at (vn)   {$n$}  ;
\node [right=2pt] at (vn')  {$n+1$};
\node [right=2pt] at (vn'') {$n+2$};

\node [below = 10pt ] at (w0) {$\frak{G}_n$};
\node [below = 10pt ] at (v0) {$\frak{H}_n$};

\end{tikzpicture}

\caption{Frames $\frak{G}_n$ and $\frak{H}_n$}
  \label{fig3}
\end{figure}

As we next observe, Theorem~\ref{thr:nat} covers countably many
logics, countably many pairs of which are incompatible.  First, note
that $\mathbf{L}(\nat, <)$ and $\mathbf{L}(\nat, \leqslant)$ are
incompatible. Let
$$
\begin{array}{lcl}
  Z & = & \Box (\Box p \imp p) \imp (\Diamond \Box p \imp \Box p);\\
  \mbox{\it ref}  & = & \Box p \imp p.
\end{array}
$$
It is well known~\cite{Goldblatt92} that
$\langle \nat, < \rangle \models Z$, but
$\langle \nat, \leqslant \rangle \not\models Z$; hence,
$\mathbf{L}(\nat, <) \not\subseteq \mathbf{L}(\nat, \leqslant)$.  It
is also clear that
$\langle \nat, \leqslant \rangle \models \mbox{\it ref}$, but
$\langle \nat, < \rangle \not\models \mbox{\it ref}$; hence,
${\bf L}(\nat, \leqslant) \not\subseteq {\bf L}(\nat, <)$.

Generalising this observation, we obtain countably many logics,
countably many pairs of which are incompatible.  Let
$$
\begin{array}{lcl}
\XBox \vp
  & =
  & (q \con \Box (\neg q \imp \vp)) \dis (\neg q \con \Box (q \imp \vp)).
\end{array}
$$

Let $\frak{G}_n$ be the irreflexive chain $0, \ldots, n - 1$, followed
by the infinite reflexive chain $n, n+1, \ldots\,$, shown in
Figure~\ref{fig3} on the left.  Dually, let $\frak{H}_n$ be the
reflexive chain $0, \ldots, n - 1$, followed by the infinite
irreflexive chain $n, n+1, \ldots\,$, shown in Figure~\ref{fig3} on
the right.

We show that ${\bf L}(\frak{G}_k) \ne {\bf L}(\frak{G}_m)$ and
${\bf L}(\frak{H}_k) \ne {\bf L}(\frak{H}_m)$ provided $k \ne m$.
Indeed, $k > m$ implies
$\mathbf{L}(\frak{G}_k) \subseteq \mathbf{L}(\frak{G}_m)$: if
$\frak{G}_m, s \not\models \vp$ then $\frak{G}_k, s \not\models \vp$
since $\frak{G}_m$ is a generated subframe of $\frak{G}_k$.  Also,
$\frak{G}_n \models \Box^n \mbox{\it ref}$, but
$\frak{G}_{n+1}, 0 \not\models \Box^n \mbox{\it ref}$.  Hence,
$\mathbf{L}(\frak{G}_k) \ne \mathbf{L}(\frak{G}_m)$ if $k \ne m$.  A similar
argument, using the formula $\XBox^n Z$ to distinguish
$\mathbf{L}(\frak{H}_{n+1})$ from $\mathbf{L} (\frak{H}_{n})$, shows that
$\mathbf{L}(\frak{H}_k) \ne \mathbf{L}(\frak{H}_m)$ if $k \ne m$.

Thus, we have infinitely many logics ${\bf L}(\frak{G}_n)$ and
infinitely many logics ${\bf L}(\frak{H}_n)$. Note that, for every
$k, m \in \nat$, logics ${\bf L}(\frak{G}_k)$ and
${\bf L}(\frak{H}_m)$ are incompatible since, for every
$k, m \in \nat$, both
$\Box^{k+m} \mbox{\it ref} \in {\bf L}(\frak{G}_k) - {\bf
  L}(\frak{H}_m)$ and
$\XBox^{k+m} Z \in {\bf L}(\frak{H}_m) - {\bf L}(\frak{G}_k)$.

From Theorem~\ref{thr:nat}, we obtain the following:

\begin{corollary}
  \label{cor:nat}
  Satisfiability for\/ ${\mathbf{L}_{c}}(\nat, \leqslant)$,
  $\mathbf{L}(\nat, <)$ and\/ ${\mathbf{L}_{c}}(\nat, <)$ is\/
  $\Sigma^1_1$-hard in languages with two individual variables, one
  monadic predicate letter and one proposition letter.
\end{corollary}

The frame $\langle \nat,<\rangle$ is isomorphic to the structure
$\langle\omega,< \rangle$, where $\omega$ is the least infinite
ordinal and $<$, as for all ordinals, is the membership relation on
$\omega$.  We next generalise Theorem~\ref{thr:nat} to logics of
frames based on infinite ordinals of a special form, which include
$\omega$.

\begin{theorem}
  \label{thr:ordinals}
  Let $\alpha = \omega \cdot m + k$, for some $m$ with
  $1 \leqslant m < \omega$ and some $k < \omega$, let $R$ be a binary
  relation on $\alpha$ between $<$ and its reflexive closure
  $\leqslant$, and let $L = \mathbf{L} (\alpha,R)$.  Then, satisfiability
  for\/ $L$ is\/ $\Sigma^1_1$-hard in languages with two individual
  variables, one monadic predicate letter and one proposition letter.
\end{theorem}

\begin{proof}
  The proof is similar to that of Theorem~\ref{thr:nat}.  We only
  comment on how to obtain an analogue of
  Lemma~\ref{lem:tiling-reduction}, an encoding of the recurrent
  tiling problem in $\mathbf{L}(\alpha, \leqslant)$.  (For the general
  case of an arbitrary relation between $<$ and $\leqslant$, we use
  $\Box^+$ instead of $\Box$.)

  Since the frame $\langle \alpha, \leqslant \rangle$ may contain a
  world that does not see another world, we need to define a variant
  of the formula $A_9$ suitable for such a situation:
  $$
  \begin{array}{lcl}
    A^{\bullet}_9 & = & \forall x\, \big( (M(x) \imp \Box (\exists y\,
                        M(y) \imp \pDiamond (\exists y\, M(y) \imp
                        P_{t_0} (x))) \big).
  \end{array}
  $$

  Let $A^{\bullet}$ be the conjunction of formulas $A_0$ through $A_8$
  from Section~\ref{sec:ref}, as well as $A^{\bullet}_9$.  We claim
  that there exists a recurrent tiling satisfying \textup{($T_1$)}
  through \textup{($T_3$)} if, and only if,
  $\langle \alpha, \leqslant \rangle \not\models \neg A^{\bullet}$.

  Assume $\frak{M}, u_0 \models A^{\bullet}$, for some model
  $\frak{M} = \langle \alpha, \leqslant, D, I \rangle$ and some world
  $u_0 \in \alpha$.  Then, there exists in $\alpha$ a last copy of
  $\omega$ that has the following property: it contains a world $w$
  marked by an element, say $a_k$, of the sequence
  $a_0 \triangleleft^{I,u_0} a_1 \triangleleft^{I,u_0} a_2
  \triangleleft^{I,u_0} \ldots\,$ of elements of $D(u_0)$ whose
  existence follows from $\frak{M}, u_0 \models A_2$.  Then, a tiling
  can be obtained from the said copy of $\omega$, similarly to the way
  it was done in the proof of the ``if'' part of
  Lemma~\ref{lem:tiling-reduction}: columns are simulated by elements
  $a_0,a_1,a_2,\ldots$ of $D(u_0)$; rows are simulated by worlds
  $w_k,w_{k+1},w_{k+2},\ldots$ such that
  \begin{itemize}
  \item $w_k = w$,
  \item $w_k < w_{k+1} < w_{k+2} < \ldots\,$,
  \item $w_{k+n}\models M(a_{k+n})$, for every $n \in \nat$;
  \end{itemize}
  and a tiling function $f\colon \nat\times\nat \to T$ is defined by
  $$
  \begin{array}{lcl}
    f(n,m) = t & \mbox{ whenever } & w_{k+m}\models P_t(a_n).
  \end{array}
  $$
  Thus defined $f$ clearly satisfies \textup{($T_1$)} and
  \textup{($T_2$)}.  Also, since
  $\frak{M}, u_0 \models A^{\bullet}_9$, the set
  $$
  \{ w_{k+m} : m \in \nat \mbox{ and } w_{k+m} \models P_{t_0} (a_0)\}
  $$
  is infinite; hence, $f$ satisfies \textup{($T_3$)} and so is a
  required function.

  For the converse, we use the first copy of $\omega$ contained in
  $\alpha$ for the satisfaction of $A^{\bullet}$: first, we define the
  interpretation of all the letters occurring in $A^{\bullet}$ on the
  said copy of $\omega$ as in the proof of the ``only if'' part of
  Lemma~\ref{lem:tiling-reduction}; second, we define the
  interpretation of letters $p$, $M$, $P_t$, for each $t \in T$, and
  $\triangleleft$ to be empty at every world not belonging to the said
  copy of $\omega$; last, we define $U$ to be an arbitrary non-empty
  subset of the domain of every world not belonging to the said copy
  of $\omega$.  Then, $A^{\bullet}$ is true at the least, with respect
  to $\leqslant$, world of $\alpha$; hence, it is satisfiable.
\end{proof}

\section{Logics of dense and continuous linear orders}
\label{sec:dense}

It is not clear whether $\Sigma^1_1$-hardness results analogous to
Theorem~\ref{thr:ordinals} can be obtained for logics of linear orders
distinct from those mentioned there.  Perhaps the most significant
logics of linear orders not covered by Theorem~\ref{thr:ordinals} are
logics of the rationals and the reals with natural partial and strict
orders, i.e.  $\mathbf{L} (\mathds{Q}, \leqslant)$,
$\mathbf{L} (\mathds{Q}, <)$, $\mathbf{L} (\mathds{R}, \leqslant)$ and
$\mathbf{L} (\mathds{R}, <)$.

The proof of Lemma~\ref{lem:tiling-reduction} does not carry over to
either $\mathbf{L} (\mathds{Q}, \leqslant)$ or
$\mathbf{L} (\mathds{R}, \leqslant)$ since we cannot ensure, given a
model of the formula $A$ based on either
$\langle \mathds{Q}, \leqslant \rangle$ or
$\langle \mathds{R}, \leqslant \rangle$, that the tiling defined as in
the proof of Lemma~\ref{lem:tiling-reduction} satisfies $(T_3)$.  In
the case of $\mathbf{L} (\mathds{Q}, \leqslant)$, no such tiling
exists: $\mathbf{L} (\mathds{Q}, \leqslant)$ is recursively
enumerable~\cite{Corsi93} and hence $\Sigma^0_1$-complete.  The case
of $\mathbf{L} (\mathds{R}, \leqslant)$ might turn out to be similar
as it is not known whether $\mathbf{L} (\mathds{R}, \leqslant)$ is
distinct from $\mathbf{L} (\mathds{Q}, \leqslant)$. (The
superintuitionistic logics of $\langle \mathds{Q}, \leqslant \rangle$
and $\langle \mathds{R}, \leqslant \rangle$
coincide~\cite[p.~701]{Skvortsov2009} and are
$\Sigma^0_1$-complete~\cite[Theorem 1]{Takano87}; on the other hand,
the superintuitionistic, and hence modal, logics of predicate frames
with constant domains over $\langle \mathds{Q}, \leqslant \rangle$ and
$\langle \mathds{R}, \leqslant \rangle$ differ~\cite[Theorem
2]{Takano87}.)

A slight modification of the proof of Lemma~\ref{lem:tiling-reduction}
shows, however, that satisfiability for
$\mathbf{L} (\mathds{Q}, \leqslant)$ and
$\mathbf{L} (\mathds{R}, \leqslant)$ is $\Pi^0_1$-hard---hence,
$\mathbf{L} (\mathds{Q}, \leqslant)$ and
$\mathbf{L} (\mathds{R}, \leqslant)$ are $\Sigma^0_1$-hard---in
languages with two variables, one monadic predicate letter and one
proposition letter: simply leaving out the argument for ($T_3$), we
obtain a reduction to satisfiability for
$\mathbf{L} (\mathds{Q}, \leqslant)$ and
$\mathbf{L} (\mathds{R}, \leqslant)$ in appropriate languages of the
$\Pi^0_1$-complete~\cite{Berger66}, \cite[Appendix A.4]{BGG97}
$\nat \times \nat$ tiling problem whose solution is required to
satisfy ($T_1$) and ($T_2$), but not ($T_3$).  We do, rather,
establish a more general result.

Define $B$ to be the conjunction of formulas $A_0$ through $A_8$
(i.e., leave out $A_9$ from the formula $A$ defined in
Section~\ref{sec:ref}).

\begin{lemma}
  \label{lem:tiling-dense}
  Let $\langle W, \leqslant \rangle$ be a partial linear order
  containing an infinite ascending chain of pairwise distinct elements
  of $W$.  Then, there exists a tiling of\/ $\nat \times \nat$
  satisfying \textup{$(T_1)$} and \textup{$(T_2)$} if, and only if,
  $\langle W, \leqslant \rangle \not\models \neg B$.
\end{lemma}

\begin{proof}
  (``if'') The proof is identical to that the ``if'' part of
  Lemma~\ref{lem:tiling-reduction}, except that we leave out the
  argument for ($T_3$).

  (``only if'') Suppose $f$ is a function satisfying ($T_1$) and
  ($T_2$).  We obtain a model based on $\langle W, \leqslant \rangle$
  satisfying $B$.

  Let $D_w = \nat \cup \{-1\}$, for every $w \in W$.  To define the
  interpretation function $I$ on $\langle W, \leqslant, D \rangle$, we
  use elements of the infinite ascending chain
  $w_0 \leqslant w_1 \leqslant w_2 \leqslant \ldots $ of worlds from
  $W$ that exists by assumption: we define $I$ so that, for every
  $k \in \nat$ and every $a, b \in \mathcal{D}$,
  $$
  \begin{array}{lcl}
  \frak{M},w_k \models a\triangleleft b
  & \leftrightharpoons &
  \mbox{$k$ is even and $b=a+1$;}
  \smallskip\\
  \frak{M},w_k \models p
  & \leftrightharpoons &
  \mbox{$k$ is odd;}
  \smallskip\\
  \frak{M},w_k \models M(a)
  & \leftrightharpoons &
  \mbox{$k=2a$;}
  \smallskip\\
  \frak{M},w_k \models P_t(a)
  & \leftrightharpoons &
  \mbox{$k=2m$ and $f(a,m) = t$, for some $m\in\nat$,}
  \end{array}
  $$
  and, for every $v \not\in \{w_i : i \in \nat \}$ and every predicate
  letter $S$ of $B$,
  $$
  I(v, S) = I(w_m, S), \mbox{ where } m = \min \{ k\in\nat : v
  \leqslant w_k \}.
  $$

  It is straightforward to check that $\frak{M}, w_0 \models B$, so we
  leave this to the reader.
\end{proof}

Using a modification of the formula $B$ obtained by replacing every
occurrence of $\Box$ by that of $\Box^+$, we can prove the following
analogue of Theorem~\ref{thr:nat} (the proof uses
Lemma~\ref{lem:tiling-dense} in the same way Theorem~\ref{thr:nat}
used Lemma~\ref{lem:tiling-reduction}):

\begin{theorem}
  \label{thr:ascending-chain}
  Let $\langle W, < \rangle$ be a strict linear order containing an
  infinite ascending chain of pairwise distinct elements of $W$.  Let
  $\leqslant$ be the reflexive closure of $<$ and $R$ a binary
  relation between $<$ and $\leqslant$.  Let $L$ be a logic such that
  $\mathbf{L}(W, R) \subseteq L \subseteq \mathbf{L}_{c} (W, R)$.
  Then, satisfiability for $L$ is $\Pi^0_1$-hard---hence, $L$ is\/
  $\Sigma^0_1$-hard---in languages with two individual variables, one
  monadic predicate letter and one proposition letter.
\end{theorem}

\begin{corollary}
  \label{cor:q-and-r}
  Logics $\mathbf{L} (\mathds{Q}, \leqslant)$,
  ${\mathbf{L}_{c}} (\mathds{Q}, \leqslant)$,
  $\mathbf{L} (\mathds{Q}, <)$, ${\mathbf{L}_{c}} (\mathds{Q}, <)$,
  $\mathbf{L} (\mathds{R}, \leqslant)$,
  ${\mathbf{L}_{c}} (\mathds{R}, \leqslant)$,
  $\mathbf{L} (\mathds{R}, <)$ and ${\mathbf{L}_{c}} (\mathds{R}, <)$
  are $\Sigma^0_1$-hard in languages with two individual variables,
  one monadic predicate letter and one proposition letter.
\end{corollary}

Well-known axiomatically defined predicate modal logics coincide with
some logics mentioned in Corollary~\ref{cor:q-and-r}.  Let
$\mathbf{K}$ be the minimal propositional normal modal logic and, for
a set of formulas $\Gamma$ and a formula $\vp$, let
$\Gamma\oplus\varphi$ be the closure of $\Gamma\cup\{\varphi\}$ under
modus ponens, necessitation and propositional substitution.  Recall
the following definitions of propositional modal logics:
$$
\begin{array}{lcl}
  \mathbf{S4.3}
  & =
  & \mathbf{K}
    \oplus \mathit{ref}
    \oplus \Box p \imp \Box \Box p
    \oplus \Box (\Box p \imp q) \dis \Box (\Box q \imp p);
    \smallskip\\
  \mathbf{K4.3.D.X}
  & =
  & \mathbf{K}
    \oplus \Box p \imp \Box\Box p
    \oplus \Box (\Box^+ p \imp q) \dis \Box (\Box^+ q \imp p)
    \oplus \Diamond \top
    \oplus \Box \Box p \imp \Box p.
\end{array}
$$
For a propositional modal logic $L$, denote by $\mathbf{Q}L$ the
minimal predicate modal logic containing $\mathbf{QCl} \cup L$.  It
follows from the definitions of $\mathbf{S4.3}$ and
$\mathbf{K4.3.D.X}$ given above that logics $\mathbf{QS4.3}$ and
$\mathbf{QK4.3.D.X}$ are finitely axiomatizable and, hence,
recursively enumerable, i.e., they are in~$\Sigma^0_1$.

It is well known~\cite{Corsi93} that
$\mathbf{QS4.3} = \mathbf{L} (\mathds{Q}, \leqslant)$ and
$\mathbf{QK4.3.D.X} = \mathbf{L} (\mathds{Q}, <)$.  We, therefore, obtain
the following result:

\begin{corollary}
  \label{cor:QS4.3}
  The logics $\mathbf{QS4.3}$ and $\mathbf{QK4.3.D.X}$ are
  $\Sigma^0_1$-complete in languages with two individual variables,
  one monadic predicate letter and one proposition letter.
\end{corollary}

Thus, $\mathbf{L} (\mathds{Q}, \leqslant)$ and
$\mathbf{L} (\mathds{Q}, <)$ are $\Sigma^0_1$-complete in languages
with two individual variables, one monadic predicate letter and one
proposition letter.

\section{Some other logics}
\label{sec:rest}

In this section, we note some corollaries of
Theorem~\ref{thr:ascending-chain} other than those mentioned in
Corollary~\ref{cor:q-and-r}.  We also note that a straightforward
modification of the proof of Theorem~\ref{thr:ascending-chain}
establishes $\Sigma^0_1$-completeness of the logic $\mathbf{QK4.3}$ in
the languages we consider.

The first corollary concerns logics of infinite ordinals: as before,
for ordinals, by $<$ and $\leqslant$ we mean, respectively, the
relation $\in$ and its reflexive closure.

\begin{corollary}
  \label{cor:inf-ordinals}
  Let $\alpha$ be an infinite ordinal.  Then,
  $\mathbf{L} (\alpha, <)$, ${\mathbf{L}_{c}} (\alpha, <)$,
  $\mathbf{L} (\alpha, \leqslant)$ and
  ${\mathbf{L}_{c}} (\alpha, \leqslant)$ are $\Sigma^0_1$-hard in
  languages with two individual variables, one monadic predicate
  letter and one proposition letter.
\end{corollary}

The second concerns logics of non-standard models of the elementary
theories of some of the structures considered thus far (the elementary
theory of the structure $\frak{A}$ is denoted by
${\sf Th}(\frak{A})$):

\begin{corollary}
  \label{cor:q-r-n-ast}
  Let $\frak{A}$ be one of the structures
  $\langle \mathds{N}, {\leqslant} \rangle$,
  $\langle \mathds{N}, {<} \rangle$,
  $\langle \mathds{Q}, {\leqslant} \rangle$,
  $\langle \mathds{Q}, {<} \rangle$,
  $\langle \mathds{R}, {\leqslant} \rangle$ and
  $\langle \mathds{R}, {<} \rangle$, and let $\frak{F}$ be a
  non-standard classical first-order model of ${\sf Th}(\frak{A})$.
  Then, $\mathbf{L}(\frak{F})$ and ${\mathbf{L}_{c}}(\frak{F})$ are
  $\Sigma^0_1$-hard in languages with two individual variables, one
  monadic predicate letter and one proposition letter.
\end{corollary}

Lastly, a slight modification of the argument of
Section~\ref{sec:dense} gives us the following result on
$\mathbf{QK4.3}$, the logic of strict linear orders~\cite{Corsi93}.
We recall that
$$
\begin{array}{lcl}
\mathbf{K4.3}
  & =
  & \mathbf{K}
    \oplus \Box p \imp \Box\Box p
    \oplus \Box (\Box^+ p \imp q) \dis \Box (\Box^+ q \imp p).
\end{array}
$$
Thus, $\mathbf{QK4.3}$ is finitely axiomatizable and, hence,
recursively enumerable.

\begin{corollary}
  \label{cor:QK4.3}
  The logic $\mathbf{QK4.3}$ is $\Sigma^0_1$-complete in languages with
  two individual variables, one monadic predicate letter and one
  proposition letter.
\end{corollary}

\begin{proof}
  One can show that the formula $A^+$ defined in
  Section~\ref{sec:discrete} is satisfiable in a model based on a
  $\mathbf{QK4.3}$-frame if, and only if, there exists a tiling of\/
  $\nat \times \nat$ satisfying \textup{$(T_1)$} and \textup{$(T_2)$}.
  We only notice that, in the proof of the ``only if'' part, we need
  to show that the model satisfying $A^+$ is infinite---this readily
  follows by $A_1$, $A_2$ and $A^+_4$.
\end{proof}

\section{Discussion}
\label{sec:discussion}

\renewcommand{\mynext}{\ocircle}

We now discuss some questions arising out of the present work.

The first question is whether our main result, Theorem~\ref{thr:nat},
can be strengthened to languages with two variables and a single
monadic predicate letter: is the ``extra'' proposition letter
necessary?

For the majority of natural predicate modal---and closely related
superintuitionistic---logics similar results have been
obtained~\cite{RSh19SL,RShJLC20a,RShJLC21a} for languages with a
single monadic predicate letter (the number of
variables---two~\cite{RSh19SL} or
three~\cite{RShJLC20a,RShJLC21a}---depends on the logic).  Those
results do not, however, cover some notable logics---the known results
for the predicate counterparts of propositional modal logics
$\mathbf{GL.3}$, $\mathbf{Grz.3}$ and $\mathbf{S5}$ involve ``extra''
proposition letters~\cite[Discussion]{RSh19SL}.  The propositional
modal logics of frames $\langle \nat, \leqslant \rangle$ and
$\langle \nat, < \rangle$---in common with $\mathbf{GL.3}$,
$\mathbf{Grz.3}$ and $\mathbf{S5}$---are NP-complete, i.e., not as
computationally hard (provided PSPSACE $\ne$ NP) as PSPACE-hard
propositional logics whose first-order counterparts are known to be
undecidable in languages with a few variables and a single monadic
predicate letter.  Whether this observation points to a genuine
connection is unclear; the hypothesis, however, seems to be worth
investigating.  It seems at least plausible that predicate logics of
$\langle \nat, \leqslant \rangle$ and $\langle \nat, < \rangle$ are
decidable in languages with two variables and a single monadic letter.

We note that a stronger result, $\Sigma^1_1$-hardness of
satisfiability for languages with two variables and a single monadic
letter, is relatively easily obtainable for the logics of the naturals
in the more expressive language containing, alongside $\Box$, the
unary operator $\mynext$ (``next'') with the truth condition
$ \frak{M}, n \models^g \mynext \vp \mbox{ if } \frak{M}, n + 1
\models^g\varphi$.  The resultant logic is a notational variant of the
first-order quantified linear time temporal logic $\mathbf{QLTL}$ with
temporal operators $\Box$ (interpreted as ``always in the future'')
and $\mynext$---even without the more expressive binary operator
``until.''  It is well known~\cite[Theorem 2]{HWZ00} that
satisfiability for $\mathbf{QLTL}(\Box, \mynext)$ is $\Sigma^1_1$-hard
in languages with two variables and only monadic predicate letters:
the proof, which inspired the proofs presented above, is a reduction
of the recurrent $\nat \times \nat$ tiling problem described in
Section~\ref{sec:ref}.  Since the number of tile types in the
recurrent tiling problem is unbounded, \cite[Theorem 2]{HWZ00}
establishes $\Sigma^1_1$-hardness for languages with an unlimited
supply of monadic predicate letters.  The proof of~\cite[Theorem
2]{HWZ00} can, however, be modified to establish $\Sigma^1_1$-hardness
for languages with a single monadic predicate letter.  Perhasps the
easiest way is to modify the formulas used in the original
proof~\cite{HWZ00} so that the satisfying model has equal-seized gaps
(non-empty sequences of worlds) between the worlds corresponding to
the tiling; the size of the gaps is proportional to the number of tile
types.  The binary letter can then be modelled as in the proof of
Lemma~\ref{lem:Kripke} above.  To model all the monadic letters with a
single one, we use the following observation.  If $f$ is a tiling
function satisfying ($T_1$) through ($T_3$), then for every
$m \in \nat$, there do not exist $t, t' \in T$ such that $t \ne t'$
and, for every $n \in \nat$, both $f(n, m) = t$ and $f(n, m) = t'$: an
entire row cannot be tiled simultaneously with tiles of two distinct
types.  Therefore, a construction similar to the one used in the proof
of Lemma~\ref{lem:monadic} above would not produce a model where two
successive worlds satisfy $\forall x\, P(x)$.  Hence, we can use
$\mynext \forall x\, P(x) \con {\mynext}{\mynext} \forall x\, P(x)$ in
place of the proposition letter $q$ in the final reduction.  This
gives us the following result:

\begin{theorem}
  \label{thr:qltl}
  Satisfiability for $\mathbf{QLTL}(\Box, \mynext)$ is
  $\Sigma^1_1$-hard in languages with two individual variables and a
  single monadic predicate letter.
\end{theorem}

The second question arising out of the present work is whether
stronger lower bounds are obtainable for the logics of the
reals---provided they are distinct from the logics of the rationals.
One approach would be to attempt to adapt techniques developed by
Reynolds and Zakharyaschev~\cite{RZ01} for proving
$\Sigma^1_1$-hardness of products of two propositional modal logics of
linearly ordered frames.  In particular, Reynolds and Zakharyaschev
establish~\cite[Theorem 6.1]{RZ01} $\Sigma^1_1$-hardness of product
logics satisfying two conditions: first, the product logic admits a
frame with infinite ascending chains along both accessibility
relations; second, the order relation associated with one of the
factor logics is Dedekind complete.  Even though this setup appears
similar to predicate modal logics of the reals, it is not immediately
clear how to apply the techniques of Reynolds and
Zakharyaschev~\cite{RZ01} in our circumstances since it is not obvious
how an infinite linear partial or strict order, which has to be
transitive, can be defined on domains of a predicate Kripke model
using formulas with only two variables.

The third question is whether our results are tight.  We are not aware
of upper-bound results for logics considered here, the exception being
the logics of the rationals, which are, as mentioned in
Section~\ref{sec:dense}, $\Sigma^0_1$-complete.  Thus, a search for
upper bounds appears to be an interesting topic of future study.

The final question we mention is whether analogous results can be
obtained for superintuitionistic logic of
$\langle \nat, \leqslant \rangle$.  Whether this can be done is
unclear to us: the techniques used here appear unsuitable for
superintuitionistic logics given the difficulty of modelling the
changing values of tile types on a linear frame with a hereditary
valuation.

\section*{Acknowledgements}

We are grateful to Valentin Shehtman for discussions of an earlier
version of the paper.  We are indebted to the anonymous reviewers for
helping to improve the paper.


\begin{thebibliography}{10}

\bibitem{Abadi89}
Martin Abadi.
\newblock The power of temporal proofs.
\newblock {\em Theoretical Computer Science}, 65(1):35--83, 1989.

\bibitem{ANS79}
H.~Andr\'eka, I.~N\'emeti, and I.~Sain.
\newblock Completeness problems in verification of programs and program
  schemes.
\newblock In J.~Be\v{c}v\'a\v{r}, editor, {\em Mathematical Foundations of
  Computer Science 1979. MFCS 1979.}, volume~74 of {\em Lecture Notes in
  Computer Science}. Springer, 1979.

\bibitem{AD90}
Sergei Artemov and Giorgie Dzhaparidze.
\newblock Finite {K}ripke models and predicate logics of provability.
\newblock {\em The Journal of Symbolic Logic}, 55(3):1090--1098, 1990.

\bibitem{Behmann22}
H.~Behmann.
\newblock Beitr\"{a}ige z\"{u}r {A}lgebra der {L}ogik, inbesondere zum
  {E}ntscheidungsproblem.
\newblock {\em Mathematische Annalen}, 86:163--229, 1922.

\bibitem{Benthem85}
Johan~Van Benthem.
\newblock {\em Modal Logic and Classical Logic}.
\newblock Bibliopolis, 1985.

\bibitem{Berger66}
Robert Berger.
\newblock {\em The {U}ndecidability of the {D}omino {P}roblem}, volume~66 of
  {\em Memoirs of {A}{M}{S}}.
\newblock {A}{M}{S}, 1966.

\bibitem{BdeRV}
Patrick Blackburn, Maarten de~Rijke, and Yde Venema.
\newblock {\em Modal Logic}, volume~53 of {\em Cambridge Tracts in Theoretical
  Computer Science}.
\newblock Cambridge University Press, 2001.

\bibitem{BS93}
Patrick Blackburn and Edith Spaan.
\newblock A modal perspective on the computational complexity of attribute
  value grammar.
\newblock {\em Journal of Logic, Language, and Information}, 2:129--169, 1993.

\bibitem{BGG97}
Egon B\"{o}rger, Erich Gr\"adel, and Yuri Gurevich.
\newblock {\em The Classical Decision Problem}.
\newblock Springer, 1997.

\bibitem{BG07}
Torben Bra\"uner and Silvio Ghilardi.
\newblock First-order modal logic.
\newblock In Patrick Blackburn, Johan~Van Benthem, and Frank Wolter, editors,
  {\em Handbook of Modal Logic}, volume~3 of {\em Studies in Logic and
  Practical Reasoning}, pages 549--620. Elsevier, 2007.

\bibitem{ChRyb00}
Alexander Chagrov and Mikhail Rybakov.
\newblock Standard translations of non-classical formulas and relative
  decidability of logics.
\newblock In {\em Annals of the research seminar of the Center for Logic of the
  Institute of Philosophy of the Russian Academy of Sciences}, volume XIV,
  pages 81--98. Russian Academy of Sciences, 2000.
\newblock In Russian.

\bibitem{ChRyb03}
Alexander Chagrov and Mikhail Rybakov.
\newblock How many variables does one need to prove {P}{S}{P}{A}{C}{E}-hardness
  of modal logics?
\newblock In Philippe Balbiani, Nobu-Yuki Suzuki, Frank Wolter, and Michael
  Zakharyaschev, editors, {\em Advances in Modal Logic 4}, pages 71--82. King's
  College Publications, 2003.

\bibitem{Corsi93}
Giovanna Corsi.
\newblock Quantified modal logics of positive rational numbers and some related
  systems.
\newblock {\em Notre Dame Journal of Formal Logic}, 34(2):263--283, 1993.

\bibitem{Elgot61}
Calvin~C. Elgot.
\newblock Decision problems of finite automata design and related arithmetics.
\newblock {\em Transactions of the American Mathematical Society},
  98(1):21---51, 1961.

\bibitem{FM98}
Melvin Fitting and Richard~L. Mendelsohn.
\newblock {\em First-{O}rder Modal Logic}, volume 277 of {\em Synthese
  Library}.
\newblock Kluwer Academic Publishers, 1998.

\bibitem{Gabbay81}
Dov Gabbay.
\newblock {\em Semantical Investigations in Heyting's Intuitionistic Logic}.
\newblock D. Reidel, 1981.

\bibitem{GHR94}
Dov Gabbay, Ian Hodkinson, and Mark Reynolds.
\newblock {\em Temporal Logic: Mathematical Foundations and Computational
  Aspects, Volume 1}, volume~28 of {\em Oxford Logic Guides}.
\newblock Oxford University Press, 1994.

\bibitem{GKWZ}
Dov Gabbay, Agi Kurucz, Frank Wolter, and Michael Zakharyaschev.
\newblock {\em Many-{D}imensional Modal Logics: Theory and Applications},
  volume 148 of {\em Studies in Logic and the Foundations of Mathematics}.
\newblock Elsevier, 2003.

\bibitem{GSh93}
Dov Gabbay and Valentin Shehtman.
\newblock Undecidability of modal and intermediate first-order logics with two
  individual variables.
\newblock {\em The Journal of Symbolic Logic}, 58(3):800--823, 1993.

\bibitem{GSh98}
Dov Gabbay and Valentin Shehtman.
\newblock Products of modal logics, {P}art 1.
\newblock {\em Logic Journal of the {I}{G}{P}{L}}, 6(1):73--146, 1998.

\bibitem{GShS}
Dov Gabbay, Valentin Shehtman, and Dmitrij Skvortsov.
\newblock {\em Quantification in Nonclassical Logic, Volume 1}, volume 153 of
  {\em Studies in Logic and the Foundations of Mathematics}.
\newblock Elsevier, 2009.

\bibitem{Garson01}
James~W. Garson.
\newblock Quantification in modal logic.
\newblock In D.~M. Gabbay and F.~Guenthner, editors, {\em Handbook of
  Philosophical Logic, vol. 3}, pages 267--323. Springer, Dordrecht, 2001.

\bibitem{Goldblatt92}
Robert Goldblatt.
\newblock {\em Logics of Time and Computation}, volume~7 of {\em CSLI Lecture
  Notes}.
\newblock Center for the Study of Language and Information, second edition,
  1992.

\bibitem{Goldblatt11}
Robert Goldblatt.
\newblock {\em Quantifiers, Propositions and Identity: Admissible Semantics for
  Quantified Modal and Substructural Logics}.
\newblock Lecture Notes in Logic. Cambridge University Press, 2011.

\bibitem{GKV97}
Erich Gr\"{a}del, Phokion~G. Kolaitis, and Moshe~Y. Vardi.
\newblock On the decision problem for two-variable first-order logic.
\newblock {\em Bulletin of Symbolic Logic}, 3(1):53--69, 1997.

\bibitem{Halpern95}
Joseph~Y. Halpern.
\newblock The effect of bounding the number of primitive propositions and the
  depth of nesting on the complexity of modal logic.
\newblock {\em Artificial Intelligence}, 75(2):361--372, 1995.

\bibitem{Harel86}
David Harel.
\newblock Effective transformations on infinite trees, with applications to
  high undecidability, dominoes, and fairness.
\newblock {\em Journal of the {A}{C}{M}}, 33(224--248), 1986.

\bibitem{HWZ00}
Ian Hodkinson, Frank Wolter, and Michael Zakharyaschev.
\newblock Decidable fragments of first-order temporal logics.
\newblock {\em Annals of Pure and Applied Logic}, 106:85--134, 2000.

\bibitem{HC96}
G.~E. Hughes and M.~J. Cresswell.
\newblock {\em A New Introduction to Modal Logic}.
\newblock Routledge, 1996.

\bibitem{KKZ05}
Roman Kontchakov, Agi Kurucz, and Michael Zakharyaschev.
\newblock Undecidability of first-order intuitionistic and modal logics with
  two variables.
\newblock {\em Bulletin of Symbolic Logic}, 11(3):428--438, 2005.

\bibitem{Kripke62}
Saul Kripke.
\newblock The undecidability of monadic modal quantification theory.
\newblock {\em Zeitschrift f\"{u}r Matematische Logik und Grundlagen der
  Mathematik}, 8:113--116, 1962.

\bibitem{Lowenheim15}
Leopold L\"{o}wenheim.
\newblock \"{U}ber {M}\"{o}glichkeiten im {R}elativkalk\"{u}l.
\newblock {\em Mathematische Annalen}, 76(4):447--470, 1915.

\bibitem{Marx99}
Maarten Marx.
\newblock Complexity of products of modal logics.
\newblock {\em Journal of Logic and Computation}, 9(2):197--214, 1999.

\bibitem{MMO65}
Sergei Maslov, Gregory Mints, and Vladimir Orevkov.
\newblock Unsolvability in the constructive predicate calculus of certain
  classes of formulas containing only monadic predicate variables.
\newblock {\em Soviet Mathematics Doklady}, 6:918--920, 1965.

\bibitem{Merz92}
Stephan Merz.
\newblock Decidability and incompleteness results for first-order temporal
  logics of linear time.
\newblock {\em Journal of Applied Non-Classical Logics}, 2(2):139--156, 1992.

\bibitem{Mints68}
Gregory Mints.
\newblock Some calculi of modal logic.
\newblock {\em Trudy Matematicheskogo Instituta imeni {V}. {A}. {S}teklova},
  98:88--111, 1968.
\newblock in Russian.

\bibitem{Mortimer75}
Michael Mortimer.
\newblock On languages with two variables.
\newblock {\em Zeitschrift f\"ur Mathematische Logik und Grundlagen der
  Mathematik}, pages 135--140, 1975.

\bibitem{Ono77}
Hiroakira Ono.
\newblock On some intuitionistic modal logics.
\newblock {\em Publications of the Research Institute for Mathematical
  Sciences}, 13(3):687--722, 1977.

\bibitem{Prior67}
Arthur Prior.
\newblock {\em Past, Present, and Future}.
\newblock Oxford University Press, 1967.

\bibitem{RZ01}
Mark Reynolds and Michael Zakharyaschev.
\newblock On the products of linear modal logics.
\newblock {\em Journal of Logic and Computation}, 11(6):909--931, 2001.

\bibitem{Rogers}
Hartley Rogers.
\newblock {\em Theory of Recursive Functions and Effective Computability}.
\newblock McGraw-Hill, 1967.

\bibitem{RShIGPL18}
Mikhail Rybakov and Dmitry Shkatov.
\newblock Complexity and expressivity of propositional dynamic logics with
  finitely many variables.
\newblock {\em Logic Journal of the IGPL}, 26(5):539--547, 2018.

\bibitem{RShAiML18} Mikhail Rybakov and Dmitry Shkatov.  \newblock A
  recursively enumerable {K}ripke complete first-order logic not
  complete with respect to a first-order definable class of frames.
  \newblock In G.~Metcalfe G.~Bezhanishvili, G.~D'Agostino and
  T.~Studer, editors, {\em Advances in Modal Logic}, volume~12, pages
  531--540. College Publications, 2018.

\bibitem{RShICTAC18}
Mikhail Rybakov and Dmitry Shkatov.
\newblock Complexity and expressivity of branching- and alternating-time
  temporal logics with finitely many variables.
\newblock In B.~Fischer B. and T.~Uustalu, editors, {\em Theoretical Aspects of
  Computing--ICTAC 2018}, volume 11187 of {\em Lecture Notes in Computer
  Science}, pages 396--414, 2018.

\bibitem{RShIGPL19}
Mikhail Rybakov and Dmitry Shkatov.
\newblock Complexity of finite-variable fragments of propositional modal logics
  of symmetric frames.
\newblock {\em Logic Journal of the IGPL}, 27(1):60--68, 2019.

\bibitem{RSh19SL}
Mikhail Rybakov and Dmitry Shkatov.
\newblock Undecidability of first-order modal and intuitionistic logics with
  two variables and one monadic predicate letter.
\newblock {\em Studia Logica}, 107(4):695--717, 2019.

\bibitem{RSh20JLC}
Mikhail Rybakov and Dmitry Shkatov.
\newblock Recursive enumerability and elementary frame definability in
  predicate modal logic.
\newblock {\em Journal of Logic and Computation}, 30(2):549--560, 2020.

\bibitem{RSh20AiML}
Mikhail Rybakov and Dmitry Shkatov.
\newblock Algorithmic properties of first-order modal logics of the natural
  number line in restricted languages.
\newblock In Nicola Olivetti, Rineke Verbrugge, Sara Negri, and Gabriel Sandu,
  editors, {\em Advances in Modal Logic}, volume~13, pages
  523--539. College Publications,
  2020.

\bibitem{RShJLC20a}
Mikhail Rybakov and Dmitry Shkatov.
\newblock Algorithmic properties of first-order modal logics of finite {K}ripke
  frames in restricted languages.
\newblock {\em Journal of Logic and Computation}, 30(7):1305--1329, 2020.

\bibitem{RShJLC21b}
Mikhail Rybakov and Dmitry Shkatov.
\newblock Complexity of finite-variable fragments of products with {K}.
\newblock {\em Journal of Logic and Computation}, 31(2):426--443, 2021.

\bibitem{RShJLC21a}
Mikhail Rybakov and Dmitry Shkatov.
\newblock Algorithmic properties of first-order superintuitionistic logics of
  finite {K}ripke frames in restricted languages.
\newblock {\em Journal of Logic and Computation}, 31(2):494--522, 2021.

\bibitem{Shehtman19}
Valentin Shehtman and Dmitry Shkatov.
\newblock On one-variable fragments of modal predicate logics.
\newblock In {\em Proceedings of SYSMICS2019}, pages 129--132. Institute for
  Logic, Language and Computation, University of Amsterdam, 2019.

\bibitem{ShS90}
Valentin Shehtman and Dmitrij Skvortsov.
\newblock Semantics of non-classical first order predicate logics.
\newblock In P.~Petkov, editor, {\em Mathematical Logic}, pages 105--116.
  Plenum Press, New York, 1990.

\bibitem{Skvortsov2009}
Dmitrij Skvortsov.
\newblock A remark on superintuitionistic predicate logics of {K}ripke frames
  with constant and with nested domains.
\newblock {\em Journal of Logic and Computation}, 21(4):697--713, 2009.

\bibitem{SSh93}
Dmitrij Skvortsov and Valentin Shehtman.
\newblock Maximal {K}ripke-type semantics for modal and superintuitionistic
  predicate logics.
\newblock {\em Annals of Pure and Applied Logic}, 63(1):69--101, 1993.

\bibitem{Spaan93}
Edith Spaan.
\newblock {\em Complexity of Modal Logics}.
\newblock PhD thesis, University of Amsterdam, 1993.

\bibitem{Szalas86}
Andrzej Szalas.
\newblock Concerning the semantic consequence relation in first-order temporal
  logic.
\newblock {\em Theoretical Computer Science}, 47:329--334, 1986.

\bibitem{SzalasHolenderski88}
Andrzej Szalas and Leszek Holenderski.
\newblock Incompleteness of first-order temporal logic with until.
\newblock {\em Theoretical Computer Science}, 57:317--325, 1988.

\bibitem{Takano87}
Mitio Takano.
\newblock Ordered sets {R} and {Q} as bases of {K}ripke models.
\newblock {\em Studia Logica}, 46:137--148, 1987.

\bibitem{Wolter00}
Frank Wolter.
\newblock First order common knowledge logics.
\newblock {\em Studia Logica}, 65:249--271, 2000.

\bibitem{WZ01}
Frank Wolter and Michael Zakharyaschev.
\newblock Decidable fragments of first-order modal logics.
\newblock {\em The Journal of Symbolic Logic}, 66:1415--1438, 2001.

\end{thebibliography}

\end{document}